\documentclass[conference]{IEEEtran}
\usepackage[latin9]{inputenc}
\usepackage{float}
\usepackage{amsmath}
\usepackage{amssymb}
\usepackage{stmaryrd}
\usepackage{graphicx}
\usepackage{amsthm}

\makeatletter

\floatstyle{ruled}
\newfloat{algorithm}{tbp}{loa}
\providecommand{\algorithmname}{Algorithm}
\floatname{algorithm}{\protect\algorithmname}

\IEEEoverridecommandlockouts
\usepackage{cite}
\usepackage{amsfonts}
\usepackage{textcomp}
\usepackage{xcolor}
\usepackage{epstopdf}

\def\BibTeX{{\rm B\kern-.05em{\sc i\kern-.025em b}\kern-.08em
    T\kern-.1667em\lower.7ex\hbox{E}\kern-.125emX}}

\theoremstyle{plain}
\newtheorem{theorem}{\protect\theoremname}
\theoremstyle{remark}

\theoremstyle{plain}

\theoremstyle{plain}

\usepackage[english]{babel}
\providecommand{\corollaryname}{Corollary}
\providecommand{\remarkname}{Remark}
\providecommand{\theoremname}{Theorem}
\providecommand{\lemmaname}{Lemma}

\makeatother

\begin{document}

\title{Nested Hybrid Cylindrical Array Design and DoA Estimation for Massive IoT Networks}

\author{Zhipeng Lin, Tiejun Lv, \emph{Senior Member, IEEE}, Wei Ni, \emph{Senior
Member, IEEE}, J. Andrew Zhang,\\
 \emph{Senior Member, IEEE}, and Ren Ping Liu, \emph{Senior Member,
IEEE}
\thanks{
This work was supported in part by the National
Natural Science Foundation of China (NSFC) under Grants 61671072, and in part by the Beijing Natural Science Foundation under Grants L192025.
\emph{(Corresponding author: Tiejun Lv.)}

Z. Lin and T. Lv are with the School of Information and Communication
Engineering, BUPT, Beijing, China (email: \{linlzp, lvtiejun\}@bupt.edu.cn).
Z. Lin is also with the School of Electrical and Data Engineering,
UTS, Sydney, Australia.

W. Ni is with the Data61, CSIRO, Sydney, Australia (e-mail: Wei.Ni@data61.csiro.au).

J. A. Zhang and R. P. Liu are with the School of Electrical and Data
Engineering, UTS, Sydney, Australia (e-mail: \{Andrew.Zhang, RenPing.Liu\}@uts.edu.au).}}
\maketitle
\begin{abstract}
Reducing cost and power consumption while maintaining high
network access capability  is a key physical-layer requirement of massive Internet
of Things (mIoT) networks. Deploying a hybrid array is a cost- and
energy-efficient way to meet the requirement, but would penalize system
degree of freedom (DoF) and channel estimation accuracy. This is because
signals from multiple antennas are combined by a radio frequency (RF)
network of the hybrid array. This paper presents a novel hybrid uniform
circular cylindrical array (UCyA) for mIoT networks. We design a nested
hybrid beamforming structure based on sparse array techniques and
propose the corresponding channel estimation method based on the second-order
channel statistics. As a result, only a small number of RF chains
are required to preserve the DoF of the UCyA. We also propose a new
tensor-based two-dimensional (2-D) direction-of-arrival (DoA) estimation
algorithm tailored for the proposed hybrid array. The algorithm suppresses the noise components in all tensor modes and
operates on the signal data model directly, hence improving estimation
accuracy with an affordable computational complexity. Corroborated
by a  Cram\'{e}r-Rao lower bound (CRLB) analysis, simulation results
show that the proposed hybrid UCyA array and the DoA estimation algorithm
can accurately estimate the 2-D DoAs of a large number of IoT devices.
\end{abstract}

\begin{IEEEkeywords}
Massive IoT, massive MIMO, hybrid beamformer, sparse array, tensor.
\end{IEEEkeywords}

\section{Introduction }

\global\long\def\figurename{Fig.}
Low-cost, low-power millimeter-wave (mmWave) techniques
have been developed to provide radio access capacity for massive Internet
of Things (mIoT) applications, such as smart city infrastructure,
healthcare and self-driving cars \cite{re1,re2,re3}.  In an mIoT
network, a large number of IoT devices are connected to an Internet-enabled
system \cite{JunZou,LIN}. Combined with advanced multiple access
techniques, mmWave massive multiple-input multiple-output (MIMO) can
significantly increase network capacity and can be potentially applied
to mIoT networks \cite{MMIot}. However, high hardware cost and power
consumption are two major obstacles of applying mmWave massive MIMO
into mIoT networks \cite{LIN}. It is unrealistic to provide a radio-frequency
(RF) chain for each antenna, as fully digital beamforming techniques
would require \cite{ya1}. Hybrid beamforming is an appropriate architecture
in which a low-dimensional digital beamforming in the baseband and
a high-dimensional analog beamforming at the RF front-end are used
\cite{andrew,ya2}.  Most conventional channel estimation
schemes for hybrid beamforming have been designed with given channel
information \cite{ya2}. Some of them apply RF networks to directly combine the received
signals from multiple antennas, resulting in resolution losses of
channel estimation accuracy \cite{re4Survey1,re4Survey2}. As a result,
the system degree of freedom (DoF), measuring the number of targets
which can be sensed and estimated at the base station (BS) \cite{Nested},
would decrease.

To increase the system DoF with a limited number of antennas, the
concept of sparse array, such as minimum redundancy array (MRA) \cite{MRA},
minimum hole array (MHA) \cite{MHA}, nested array \cite{Nested},
and coprime array \cite{coprime}, has attracted considerable attention.
By exploiting the second-order statistics of impinging signals, these
sparse arrays are capable of identifying $O(N^{2})$ uncorrelated
sources with only $N$ physical elements. However, existing sparse
array techniques have been typically used to design linear or square
arrays. Compared to square arrays, circular arrays have a much more
compact size, less sensitivity to mutual coupling, and inherently
more symmetric structure \cite{Circular1}, and hence, they are more suitable
for mIoT applications.

Channel estimation is challenging for sparse arrays. There have been
attempts to apply the celebrated multiple signal classification (MUSIC)
algorithm to networks equipped with sparse arrays \cite{Nested,Nested2D1,Nested2D2,coprime,we11}.
But the estimation accuracy of this algorithm is unsatisfactory, depending
on the searching step and signal correlation. Tensor-based multi-dimensional
MUSIC algorithms were proposed in \cite{tensorMUSIC,Tensor2DMUSIC}
for sparse arrays to improve estimation accuracy. However, since the
MUSIC spectrum of their algorithms is a product of multiple separable
second-order spectra, undesirable \textit{cross-terms} \cite{Tensor2DMUSIC}
would arise, leading to incorrect spectral peak search results. To
solve this problem, CANDECOMP/PARAFAC (CP)-based tensor channel estimation
algorithms were proposed \cite{tensormmWave}. However, these algorithms
have a very high computational complexity.

 This paper proposes a new nested massive hybrid uniform
circular cylindrical array (UCyA) design and the corresponding tensor-based
angle estimation algorithm for the uplink of mIoT networks. By exploiting
the sparse array techniques, the proposed hybrid antenna array enables
the BS to estimate the DoAs of a large number of devices with much
fewer RF chains than antennas. As a result, the massive access requirement
of mIoT can be met, with significantly reduced hardware cost and network
overhead. Authors of \cite{NestedUCA} and \cite{NestedUCA2} proposed
nested sparse circular arrays for direction-of-arrival (DoA) estimation.
However, they directly computed the autocorrelation of impinging signals,
which unfortunately destructed the original symmetric structures of
circular arrays and penalized the channel estimation accuracy significantly.
Different from \cite{NestedUCA} and \cite{NestedUCA2}, we transform
the nonlinear phase of the UCyA steering vectors to be linear to the
element locations, so that the horizontal symmetric structure of UCyA
can be preserved. In addition, since the DoA estimation algorithms
developed in \cite{NestedUCA} and \cite{NestedUCA2}   were
matrix-based, the estimation accuracy gap between the algorithms and
the CRLB is large and the algorithms cannot be directly applied to
high-dimension DoA estimation. For our new hybrid UCyA array design,
we propose a new tensor $n$-rank enhancement method and a new tensor-based
two-dimensional (2-D) DoA estimation algorithm. The algorithm suppresses
the noise components in each mode of the signal tensor model. As a
result, the DoAs of a large number of IoT devices can be accurately
estimated with a much smaller number of RF chains.  The key contributions
of this paper are summarized as follows.
\begin{itemize}
\item We design a new nested hybrid UCyA, which reduces the required number
of RF chains while preserving the inherently horizontal symmetric
structure of the UCyA to maintain a good channel estimation accuracy.
The theory of phase-space transformation is first used to transform
the nonlinear phase of the UCyA steering vectors to be linear to the
element locations. Then, we design the RF-chain connection network
by exploiting the sparse array technique, and utilize its generated
difference coarray for parameter estimation.
\item We analyze the rank relationship between signal matrix and the signal
tensor model in each dimension, and propose a tensor $n$-rank enhancement
method which ensures that the signal and noise subspaces can be properly
decomposed in all dimensions.
\item We propose a new tensor-based two-dimensional (2-D) DoA estimation
algorithm, based on our hybrid array design. We combine the tensor
tool with the estimation of signal parameters via rotational invariance
technique (ESPRIT) to estimate the elevation angles. Then, we substitute
the estimates to derive the azimuth angles by using tensor MUSIC.
Simulation results show that, by suppressing the noise components
in all tensor modes, the proposed algorithm can significantly improve
the estimation accuracy, as compared to the state of the art.
\end{itemize}

The rest of this paper is organized as follows. The system model is
introduced in Sections II. In Sections III and IV, we design the hybrid
beamformer and propose the new spatial smoothing-based method to enhance
the $n$-rank of measurement tensor. In Section V, we design a new
tensor-based 2-D DoA estimation algorithm, and analyze the system
complexity. In Section VII, simulation results are presented, followed
by conclusions in Section VIII.

\textbf{Preliminary and notation:} We
provide a brief introduction of tensor and the notations used in this
paper. Tensor is the generalization of scalar (which has a zero-order
mode), vector (which has an one-order mode), and matrix (which has
two-order modes) to arrays which have an arbitrary order of modes.
We use $\mathcal{A}\in\mathbb{C}^{I_{1}\times I_{2}\times\cdots\times I_{N}}$
to denote an order-$N$ tensor, whose elements (entries) are $a_{i_{1},i_{2},\cdots,i_{N}},$
$i_{n}=1,2,\ldots,I_{n}$, and the indices in the $n$-th mode of
$\mathcal{A}$ range from 1 to $I_{n}$.

In this paper, we use the following notations and
operations in accordance with \cite{TensorDecompositions1}.
\begin{itemize}
\item $a$, $\mathbf{a}$ and $\mathbf{A}$ stand for a
scalar, a column vector, and a matrix, respectively; $\mathbf{I}_{K}$
and $\mathbf{0}_{M\times K}$ denote a $K\times K$ identity matrix
and an $M\times K$ zero matrix, respectively; $\mathbf{A}^{\ast}$,
$\mathbf{A}^{T}$ and $\mathbf{A}^{H}$ denote the conjugate, transpose
and conjugate transpose of $\mathbf{A}$, respectively; $\left\Vert \mathbf{A}\right\Vert _{\textrm{F}}$
denotes the Frobenius norm of $\mathbf{A}$; $\otimes$ and $\diamond$
denote the Kronecker product and Khatri-Rao product, respectively;
and $\textrm{invec}(\cdot)$ denotes the inverse algorithm of vectorization.
\item The mode-$n$ unfolding (also known as matricization)
of a tensor $\mathcal{A}\in\mathbb{C}^{I_{1}\times I_{2}\times\cdots\times I_{N}}$,
denoted by $\mathbf{A}_{(n)}\in\mathbb{C}^{I_{n}\times(I_{1}I_{2}\cdots I_{N}/I_{n})}$,
arranges the fibers in the $n$-th mode of $\mathcal{A}$ as the columns
of the resulting matrix $\mathbf{A}_{(n)}$. The $n$-rank of $\mathcal{Y}$,
denoted by $\textrm{Rank}_{n}(\mathcal{A})$, is the rank of the mode-$n$
 unfolding of tensor $\mathcal{A}$.
\item The $n$-mode product of a tensor $\mathcal{A}\in\mathbb{C}^{I_{1}\times I_{2}\times\cdots\times I_{N}}$
and a matrix $\mathbf{B}\in\mathbb{C}^{J_{n}\times I_{n}}$ is defined
as
\begin{equation}
\mathcal{C}=\mathcal{A}\times_{n}\mathbf{B}\in\mathbb{C}^{I_{1}\times\cdots\times I_{n-1}\times J_{n}\times I_{n+1}\times\cdots\times I_{N}},\label{eq:nmodeproduct-1}
\end{equation}
which can be written in the form of the mode-$n$ matricized tensor:
$\mathbf{C}_{(n)}=\mathbf{B}\mathbf{A}_{(n)}.$
\end{itemize} \begin{itemize}\setlength{\itemsep}{-10pt }
\item The multilinear product of a tensor $\mathcal{A}\in\mathbb{C}^{I_{1}\times I_{2}\times\cdots\times I_{N}}$
and matrices $\mathbf{B}^{(n)}\in\mathbb{C}^{J_{n}\times I_{n}}$,
$n=1,2,\ldots,N$, is a sequence of contractions with each being an
$n$-mode product, i.e.,
\begin{equation}
\mathcal{C}=\mathcal{A}\times_{1}\mathbf{B}^{(1)}\times_{2}\mathbf{B}^{(2)}\cdots\times_{N}\mathbf{B}^{(N)}\in\mathbb{C}^{J_{1}\times J_{2}\times\cdots\times J_{N}},\label{eq:multiproduct-1}
\end{equation}
which can be equivalently expressed as $\mathcal{C}=\left\llbracket \mathcal{A};\mathbf{B}^{(1)},\mathbf{B}^{(2)},\ldots,\mathbf{B}^{(N)}\right\rrbracket .$
The mode-$n$ unfolding (or matricization) of $\mathcal{C}$ is given
by
\begin{align}
\mathbf{C}_{(n)} & =\mathbf{B}^{(n)}\mathbf{A}_{(n)}(\mathbf{B}^{(n+1)}\otimes\mathbf{B}^{(n+2)}\otimes\nonumber \\
 & \cdots\otimes\mathbf{B}^{(N)}\otimes\mathbf{B}^{(1)}\otimes\mathbf{B}^{(2)}\otimes\cdots\otimes\mathbf{B}^{(n-1)})^{T}.\label{eq:propertymultiproduct-1}
\end{align}
\item The outer product of two tensors $\mathcal{A}\in\mathbb{C}^{I_{1}\times I_{2}\times\cdots\times I_{N}}$
and $\mathcal{B}\in\mathbb{C}^{J_{1}\times J_{2}\times\cdots\times J_{M}}$
is given by
\begin{equation}
\mathcal{C}=\mathcal{A}\circ\mathcal{B}\in\mathbb{C}^{I_{1}\times I_{2}\times\cdots\times I_{N}\times J_{1}\times J_{2}\times\cdots\times J_{M}},
\end{equation}
whose elements are $c_{i_{1},i_{2},\cdots,i_{N},j_{1},j_{2},\cdots,j_{M}}=a_{i_{1},i_{2},\cdots,i_{N}}\cdot b_{j_{1},j_{2},\cdots,j_{M}}.$
\end{itemize}
\begin{itemize}\setlength{\itemsep}{-10pt }
\item Two tensors, $\mathcal{A}\in\mathbb{C}^{I_{1}\times I_{2}\times\cdots\times I_{N}}$
and $\mathcal{B}\in\mathbb{C}^{I_{1}\times\cdots\times I_{n-1}\times J_{n}\times I_{n+1}\times\cdots\times I_{N}}$,
can be concatenated in their $n$-th mode, as given by
\begin{equation}
\mathcal{C}=\left[\mathcal{A}\sqcup_{n}\mathcal{B}\right]\in\mathbb{C}^{I_{1}\times\cdots\times I_{n-1}\times\left(I_{n}+J_{n}\right)\times I_{n+1}\times\cdots\times I_{N}}.
\end{equation}
\item The Tucker decomposition decomposes a tensor $\mathcal{A}\in\mathbb{C}^{I_{1}\times I_{2}\times\cdots\times I_{N}}$
into a core tensor $\mathcal{G}\in\mathbb{C}^{R_{1}\times R_{2}\times\cdots\times R_{N}}$
multiplied by a factor matrix $\mathbf{C}^{(n)}=\left[\mathbf{c}_{r_{n}=1}^{(n)},\mathbf{c}_{r_{n}=2}^{(n)},\ldots,\mathbf{c}_{r_{n}=R_{n}}^{(n)}\right]\in\mathbb{C}^{I_{n}\times R_{n}}$ ($\mathbf{c}_{r_{n}}^{(n)}\in\mathbb{C}^{I_{n}\times1}$  and
$n=1,2,\ldots,N$) in each mode, i.e.,
\begin{align}
\mathcal{A} & =\sum_{r_{1}=1}^{R_{1}}\sum_{r_{2}=1}^{R_{2}}\cdots\sum_{r_{N}=1}^{R_{N}}g_{r_{1}r_{2}\cdots r_{N}}\left(\mathbf{c}_{r_{1}}^{(1)}\circ\mathbf{c}_{r_{2}}^{(2)}\circ\cdots\mathbf{c}_{r_{N}}^{(N)}\right)\nonumber \\
 & =\left\llbracket \mathcal{G};\mathbf{C}^{(1)},\mathbf{C}^{(2)},\ldots,\mathbf{C}^{(N)}\right\rrbracket .\label{eq:Tucker}
\end{align}
The higher-order singular value decomposition (HOSVD) is a special
case of the Tucker decomposition, where the core tensor is all-orthogonal
\cite{TensorDecompositions1}, and the factor matrices are the unitary
left singular matrices of the mode-$n$ unfolding of $\mathcal{A}$.
\end{itemize}

\section{System Model}

\begin{figure}
\protect\begin{centering}
\protect\includegraphics[width=4cm]{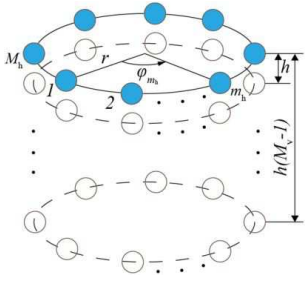}\protect
\par\end{centering}
\protect\centering{}\protect\caption{The geometric model of the UCyA.\label{fig:The-geometric-model}}
\protect
\end{figure}

As shown in Fig. \ref{fig:The-geometric-model},
a BS is equipped with an $M_{\textrm{bs}}$-antenna large-scale hybrid
mmWave UCyA consisting of $M_{\textrm{v}}$ vertically placed uniform
circular arrays (UCAs). Each of the UCAs is on a
horizontal plane with $M_{\textrm{h}}$ elements, and $M_{\textrm{bs}}=M_{\textrm{v}}M_{\textrm{h}}$.
Let $r$ be the radius of the UCyA, and $h$ be the vertical spacing
between any two adjacent vertical elements. We assume
that there are $K$ IoT devices, each equipped with a single antenna\footnote{The proposed technique can be readily applied when
multiple antennas are deployed at a device. In that case, the paths
originating from different antennas can be distinguished by transmitting
different pilot signals.}. Each device has a dominating path and different devices have separable and resolvable paths. Hence, $K$ signal paths are received by the BS.
The received signal sample at the $m_{\textrm{t}}$-th time frame
$(m_{\textrm{t}}=1,\ldots,M_{\textrm{t}})$ can be expressed as \cite{DoAModel}
\begin{align}
\mathbf{x}_{m_{\textrm{t}}}=\sum_{^{k=1}}^{K}s_{m_{\textrm{t}},k}\mathbf{B}^{H}\mathbf{a}_{\textrm{bs}}(\phi_{k},\theta_{k})+\mathbf{n}_{m_{\textrm{t}}},\label{eq:ori_signalmodel}
\end{align}
where $\phi_{k}$ and $\theta_{k}$ are the azimuth and elevation
DoAs of the $k$-th device, respectively; $\mathbf{a}_{\textrm{bs}}(\phi_{k},\theta_{k})\in\mathbb{C}^{M_{\textrm{bs}}\times1}$
denotes the steering vector of the hybrid UCyA; $s_{m_{\textrm{t}},k}$
is the received symbol of the $k$-th device at the $m_{\textrm{t}}$-th
time frame; $\mathbf{n}_{m_{\textrm{t}}}\in\mathbb{C}^{M_{\textrm{bs}}\times1}$
denotes the additive white Gaussian noise (AWGN); $\mathbf{B}\in\mathbb{C}^{M_{\textrm{bs}}\times M_{\textrm{bsd}}}$
is the hybrid beamforming matrix; and $M_{\textrm{bsd}}$
is the number of data streams.

Given the structure of the UCyA, the array steering vector $\mathbf{a}_{\textrm{bs}}(\phi_{k},\theta_{k})$
can be written as $\mathbf{a}_{\textrm{bs}}(\phi_{k},\theta_{k})=\mathbf{a}_{\textrm{v}}(\theta_{k})\otimes\mathbf{a}_{\textrm{h}}(\theta_{k},\phi_{k}),$
where $\mathbf{a}_{\textrm{v}}(\theta_{k})$ and $\mathbf{a}_{\textrm{h}}(\theta_{k},\phi_{k})$
are the vertical and horizontal array steering vectors with their
elements given by
\begin{align}
 & \left[\mathbf{a}_{\textrm{v}}(\theta_{k})\right]_{m_{\textrm{v}}}=a_{\textrm{v},m_{\textrm{v}}}(\theta_{k})\nonumber \\
 & \quad=\frac{1}{\sqrt{M_{\textrm{v}}}}\exp\left(-j\frac{2\pi}{\lambda}h(m_{\textrm{v}}-1)\cos(\theta_{k})\right),\label{eq:verticalarray}
\end{align}
\begin{align}
 & \left[\mathbf{a}_{\textrm{h}}(\theta_{k},\phi_{k})\right]_{m_{\textrm{h}}}=a_{\textrm{h},m_{\textrm{h}}}(\theta_{k},\phi_{k})\nonumber \\
 & \quad=\frac{1}{\sqrt{M_{\textrm{h}}}}\exp\left(j\frac{2\pi}{\lambda}r\sin(\theta_{k})\cos(\phi_{k}-\varphi_{m_{\textrm{h}}})\right),\label{eq:horizonarray}
\end{align}
where $\lambda$ is the wavelength, $m_{\textrm{v}}=1,\ldots,M_{\textrm{v}}$
and $m_{\textrm{h}}=1,\ldots,M_{\textrm{h}}$. $\varphi_{m_{\textrm{h}}}=2\pi(m_{\textrm{h}}-1)/M_{\textrm{h}}$
is the difference of the central angles between the $m_{\textrm{h}}$-th
antenna and the first antenna of each UCA. The geometric model of
the UCyA is shown in Fig. \ref{fig:The-geometric-model}. In this
paper, the antenna array can be reasonably treated as a phased array
because the signal bandwidth $B$ is much smaller than the carrier
frequency $f$, i.e. $B\ll f$\footnote{Much smaller is defined by $\left|(f\pm B)/f\right|\approx1$. When
$f=60$ GHz and $B\leq2$ GHz, it has $\left|(f\pm B)/f\right|\in[0.97,\:1.03]$.}, and the signals are narrowband.

\section{Proposed Nested 3D Hybrid UCyA }

In this section, we design the hybrid beamformer $\mathbf{B}$ for performing channel estimation.
$\mathbf{B}=\mathbf{B}_{\textrm{rf}}\mathbf{B}_{\textrm{bb}}\in\mathbb{C}^{M_{\textrm{bs}}\times M_{\textrm{bsd}}}$ can be decoupled between an analog beamforming matrix $\mathbf{B}_{\textrm{rf}}\in\mathbb{C}^{M_{\textrm{bs}}\times M_{\textrm{rf}}}$
and a digital beamforming matrix  $\mathbf{B}_{\textrm{bb}}\in\mathbb{C}^{M_{\textrm{rf}}\times M_{\textrm{bsd}}}$.
Here, $M_{\textrm{rf}}$ is the number of RF chains. We first briefly
review the concept of difference coarray and sparse array, which are
heavily used in this paper. Then, we introduce the $\mathbf{B}$ design
process in detail.

\subsection{Review of Sparse Arrays}

\noindent \textbf{Definition 1 (Difference Coarray): }For an antenna
array with $N$ elements, $\mathbf{w}_{n}$ is the position of its
$n$-th element, $n=1,2,\ldots,N$. Let $\mathbf{w}_{n}\in\mathbb{C}^{3\times1}$
denote the 3D coordinate of the $n$-th antenna array element. The locations of
all array elements are collected in the set $\mathbb{D}_{\textrm{a}}$,
i.e., $\mathbb{D}_{\textrm{a}}=\left\{\mathbf{w}_{n}\right\}$.  The
difference coarray of the antenna array is an (virtual) array with element positions given
by the set $\mathbb{D}_{\textrm{dc}}$:
\begin{figure*}
\begin{centering}
\includegraphics[width=15cm]{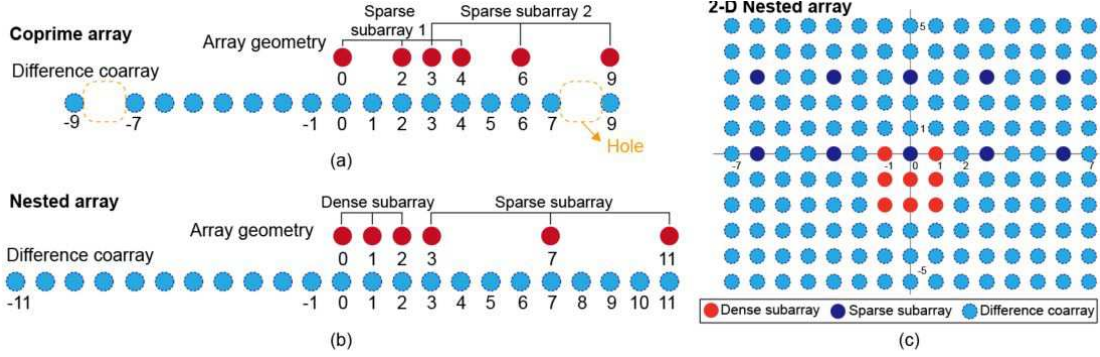}
\par\end{centering}
\centering{}\caption{An example of (a) coprime array, which is composed of two sparse subarrays:
one with $N=3$ elements and separation $M=2$, and another one with
$2M-1$ elements and separation $N$; (b) nested array, which is composed
of a dense subarray with $N_{1}=3$ elements and separation 1, and
a sparse subarray with $N_{2}=3$ elements and separation $N_{1}+1$;
and (c) 2-D nested array, which is composed of a $3\times3$ dense
subarray and a $5\times2$ sparse subarray. \label{fig:NestedCoprime}}
\end{figure*}
\begin{equation}
\mathbb{D}_{\textrm{dc}}=\left\{ \mathbf{w}_{n_{1}}-\mathbf{w}_{n_{2}}\right\} ,\;\forall n_{1},n_{2}=1,2,\ldots,N.\label{eq:DiffCoarray}
\end{equation}
According to \eqref{eq:DiffCoarray}, the element positions of the
difference coarray are the \textit{(self) differences} between the
locations of original physical antenna elements.

Based on the definition of difference coarray, we
can define a cross difference coarray, which corresponds to the cross
differences between the element locations of two arrays with $N$
and $M$ elements:
\begin{equation}
\mathbb{D}_{\textrm{cdc}}=\pm\left\{ \mathbf{w}_{n}-\mathbf{w}_{m}\right\} ,\;\forall n=1,2,\ldots,N,\:m=1,2,\ldots,M.\label{eq:crosscoarray}
\end{equation}
According to \eqref{eq:DiffCoarray} and \eqref{eq:crosscoarray},
we can see that the concept of difference coarray arises naturally
in the second-order statistics of the impinging signals.
For example, we consider that a signal $\mathbf{x}_{\mathbf{w}_{n}}\in\mathbb{C}^{N\times1}$
is received at the $n$-th element of an antenna array.
The
cross-correlation between the signals received at the $n_{1}$-th
and $n_{2}$-th elements of the array is given by
\begin{equation}
\textrm{E}\left\{ \mathbf{x}_{\mathbf{w}_{n_{1}}}\mathbf{x}_{\mathbf{w}_{n_{2}}}^{H}\right\} =\mathbf{R}_{(\mathbf{w}_{n_{1}}-\mathbf{w}_{n_{2}})}\in\mathbb{C}^{N\times N},\:\mathbf{w}_{n_{1}},\mathbf{w}_{n_{2}}\in\mathbb{D}_{\textrm{a}}\label{eq:DifCoarray}
\end{equation}
where $\mathbf{R}_{(\mathbf{w}_{n_{1}}-\mathbf{w}_{n_{2}})}$ can
be viewed as a signal sample received by a (larger) difference coarray
with virtual array elements located at $(\mathbf{w}_{n_{1}}-\mathbf{w}_{n_{2}})\in\mathbb{D}_{\textrm{dc}}$
\cite{Nested}.

By adequately designing the element locations, i.e., $\mathbb{D}_{\textrm{a}}$,
we can increase the number of virtual elements in the difference coarray
after computing the autocorrelation. If we use the samples from the
difference coarray to perform spectral estimation, the parameters
of much more targets can be estimated.

We proceed to introduce the concept of sparse array. An array is said
to be sparse if the spacing between a majority or all of adjacent
elements is more than one (half-wavelength) \cite{Nested,coprime}.
By applying the concept of sparse array to antenna
design, we can significantly improve the number of distinguishable  targets
using a small number of physical antenna elements \cite{Rebook}.
Some well-known 1-D sparse arrays include MRA \cite{MRA}, MHA \cite{MHA},
nested arrays \cite{Nested}, and coprime arrays \cite{coprime}.
With $O(N)$ physical array elements, both MRA and MHA can construct
difference coarrays with the size of $O(N^{2})$. However, their geometries
need to be constructed by using searching algorithms, e.g., integer
programming \cite{MHA,2DSPARSE}. Nested and coprime arrays were proposed
in \cite{Nested,coprime} with closed-form expressions for element
locations, and both of them can construct difference coarrays with
the same DoF as MRA and MHA. An example of nested
and coprime arrays, and their difference coarrays are shown in Figs.
\ref{fig:NestedCoprime}(a) and \ref{fig:NestedCoprime}(b). Nested
arrays can offer larger difference coarray DoF than coprime arrays,
as shown in Fig. \ref{fig:NestedCoprime}(b), where both of them have
six physical elements. In addition, the difference coarrays of nested
arrays consist of evenly spaced virtual elements with no holes, so
that the subspace-based estimation algorithms, such as MUSIC and ESPRIT,
can be utilized on the coarray domain without creating ambiguities
\cite{2DSPARSE}. For the details of these arrays, interested readers
can refer to \cite{MRA,MHA,Nested,coprime}.

In the next subsection, we design the phase shifter matrix $\mathbf{B}_{\textrm{ps}}$.
$\mathbf{B}_{\textrm{ps}}$ can transform the UCyA steering vectors
from the element space into a phase space, where the phases of the
array steering vectors are linear to the element locations. From \eqref{eq:DifCoarray},
we see that if we want to construct a difference coarray with a similar
geometry to that of the original array, e.g., the UCyA in our system,
the phase of the array steering vectors should vary linearly with
the element locations. However, due to the special geometry of the
UCyA, if we directly calculate the cross correlation of the array
steering vectors, it would generate a virtual non-UCyA composed of
multiple non-UCAs \cite{NestedUCA2}, leading to an increased computational
complexity and degraded the estimation accuracy.

\subsection{Phase-Space Transformation}

The analog beamforming matrix $\mathbf{B}_{\textrm{rf}}=\mathbf{B}_{\textrm{ps}}\mathbf{B}_{\textrm{rfc}}\in\mathbb{C}^{M_{\textrm{bs}}\times M_{\textrm{rf}}}$
is composed of a phase shifter matrix $\mathbf{B}_{\textrm{ps}}\in\mathbb{C}^{M_{\textrm{bs}}\times M_{\textrm{bsr}}}$
and an RF-chain connection matrix $\mathbf{B}_{\textrm{rfc}}\in\mathbb{C}^{M_{\textrm{bsr}}\times M_{\textrm{rf}}}$,
where $M_{\textrm{bsr}}$ is the number of output ports of the phase-shifter matrix. An illustration
of the RF front-end structure is shown in Fig. \ref{fig:The-block-diagram}.
Here, we design $\mathbf{B}_{\textrm{ps}}$ based on circular phase-space
transformation \cite{book11}, to transform the nonlinear phase of
UCyA steering vectors to be linear to the element locations.

\begin{figure}
\begin{centering}
\includegraphics[width=7cm]{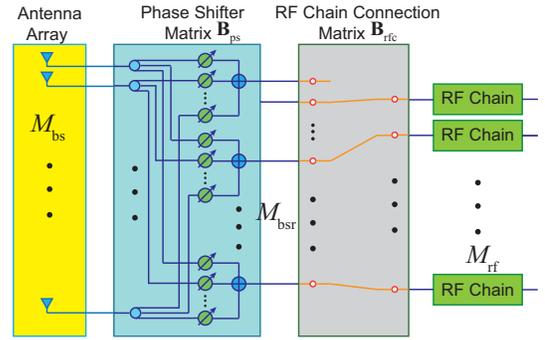}
\par\end{centering}
\centering{}\caption{The block diagram of RF front-end structure.\label{fig:The-block-diagram}}
\end{figure}

We decouple $\mathbf{B}_{\textrm{\textrm{ps}}}$ between the vertical
and horizontal planes, i.e., $\mathbf{B}_{\textrm{ps}}=\mathbf{B}_{\textrm{vps}}\otimes\mathbf{B}_{\textrm{hps}}$
with $\mathbf{B}_{\textrm{vps}}\in\mathbb{C}^{M_{\textrm{v}}\times M_{\textrm{vr}}}$
and $\mathbf{B}_{\textrm{hps}}\in\mathbb{C}^{M_{\textrm{h}}\times M_{\textrm{hr}}}$,
and thus $M_{\textrm{bsr}}=M_{\textrm{vr}}M_{\textrm{hr}}$, where
$M_{\textrm{vr}}$ and $M_{\textrm{hr}}$ are the
number of the phase-shifter output ports    along the vertical and horizontal directions,
respectively. According to the phase-space transformation of UCAs
\cite{book11}, we design $\mathbf{B}_{\textrm{hps}}$ as $\left[\mathbf{B}_{\textrm{hps}}\right]_{m_{\textrm{h}},m_{\textrm{hr}}+P+1}=e^{-j\frac{2\pi(m_{\textrm{h}}-1)}{M_{\textrm{h}}}m_{\textrm{hr}}},$
where $M_{\textrm{hr}}=2P+1$, $m_{\textrm{hr}}=-P,-P+1,\ldots,P$,
and $P$ is the highest phase-space dimension. Thus, the $M_{\textrm{h}}$-dimensional
array steering vector $\mathbf{a}_{\textrm{h}}(\theta_{k},\phi_{k})$
can be transformed into a $(2P+1)$-dimensional phase space, i.e.,
$\mathbf{a}_{\textrm{hps}}(\theta_{k},\phi_{k})=\mathbf{B}_{\textrm{hps}}^{H}\mathbf{a}_{\textrm{h}}(\theta_{k},\phi_{k})\in\mathbb{C}^{(2P+1)\times1}$.
The value of the highest phase-space dimension, $P$, can be configured
based on the following theorem.

\noindent \begin{theorem} Suppose that $M_{\textrm{h}}\geq\left\lfloor 4\pi r/\lambda\right\rfloor $.
If the highest phase-space dimension, $P$, is larger than $\left\lfloor 2\pi r/\lambda\right\rfloor $
and smaller than $M_{\textrm{hr}}/2$, then the
elements in the phase-space response can be approximated
by
\begin{align}
a_{\textrm{hps},p}(\theta_{k},\phi_{k})\approx\sqrt{M_{\textrm{h}}}j^{p}J_{p}\left(\gamma(\theta_{k})\right)\exp\left(-jp\phi_{k}\right),\label{eq:Theorem1}
\end{align}
where $\gamma(\theta_{k})=2\pi r\sin(\theta_{k})/\lambda$, $p=-P,-P+1,\ldots,P$,
and $J_{p}\left(\gamma(\theta_{k})\right)$ is the Bessel function
of the first kind of order $p$.

\noindent \end{theorem} \begin{proof} See Appendix I.\end{proof}

We set $\mathbf{B}_{\textrm{vps}}=\mathbf{I}_{M_{\textrm{v}}}$
to preserve the recurrence relations among UCAs. According to Theorem
1, the array steering vectors $\mathbf{a}_{\textrm{bs}}(\theta_{k},\phi_{k})$
after the hybrid beamformer is given by
\begin{align}
 & \mathbf{a}_{\textrm{bd}}(\theta_{k},\phi_{k})=\mathbf{B}^{H}\mathbf{a}_{\textrm{bs}}(\theta_{k},\phi_{k})\nonumber \\
 & =\left(\mathbf{B}_{\textrm{rf}}\mathbf{B}_{\textrm{bb}}\right)^{H}\mathbf{a}_{\textrm{bs}}(\theta_{k},\phi_{k})\nonumber \\
 & =\left((\mathbf{B}_{\textrm{vps}}\otimes\mathbf{B}_{\textrm{hps}})\mathbf{B}_{\textrm{rfc}}\mathbf{I}_{M_{\textrm{bsr}}}\right)^{H}\mathbf{a}_{\textrm{bs}}(\theta_{k},\phi_{k})\nonumber \\
 & =\mathbf{B}_{\textrm{rfc}}^{H}\left(\mathbf{B}_{\textrm{vps}}\otimes\mathbf{B}_{\textrm{hps}}\right)^{H}\mathbf{a}_{\textrm{bs}}(\theta_{k},\phi_{k}),\label{eq:AFTERBF1}
\end{align}
where $\mathbf{B}_{\textrm{bb}}$ is a diagonal matrix used to guarantee
the power constraint \cite{Largescaleantenna}. Without loss of generality,
we set $\mathbf{B}_{\textrm{bb}}=\mathbf{I}_{M_{\textrm{bsr}}}$ in
this paper. According to two properties of the Khatri-Rao product:
$(\mathbf{A}\otimes\mathbf{B})^{H}=\mathbf{A}^{H}\otimes\mathbf{B}^{H}$
and $(\mathbf{A}\otimes\mathbf{B})(\mathbf{C}\otimes\mathbf{D})=\mathbf{AC}\otimes\mathbf{BD}$
\cite{Arraysignalprocessingconceptsandtechniques}, \eqref{eq:AFTERBF1}
can be rewritten as
\begin{align}
 & \mathbf{a}_{\textrm{bd}}(\theta_{k},\phi_{k})\nonumber \\
 & =\mathbf{B}_{\textrm{rfc}}^{H}\left(\mathbf{B}_{\textrm{vps}}^{H}\otimes\mathbf{B}_{\textrm{hps}}^{H}\right)\left(\mathbf{a}_{\textrm{v}}(\theta_{k})\otimes\mathbf{a}_{\textrm{h}}(\theta_{k},\phi_{k})\right)\nonumber \\
 & =\mathbf{B}_{\textrm{rfc}}^{H}\left[\left(\mathbf{B}_{\textrm{vps}}^{H}\mathbf{a}_{\textrm{v}}(\theta_{k})\right)\otimes\left(\mathbf{B}_{\textrm{hps}}^{H}\mathbf{a}_{\textrm{h}}(\theta_{k},\phi_{k})\right)\right]\nonumber \\
 & =\mathbf{B}_{\textrm{rfc}}^{H}\left[\mathbf{a}_{\textrm{vps}}(\theta_{k})\otimes\mathbf{a}_{\textrm{hps}}(\theta_{k},\phi_{k})\right],\label{eq:afterBF}
\end{align}
where $\mathbf{a}_{\textrm{vps}}(\theta_{k})\in\mathbb{C}^{M_{\textrm{vr}}\times1}$,
$\mathbf{a}_{\textrm{hps}}(\theta_{k},\phi_{k})\in\mathbb{C}^{M_{\textrm{hr}}\times1}$,
$M_{\textrm{vr}}=M_{\textrm{v}}$, and $M_{\textrm{hr}}=2P+1$.

According to Theorem 1, we have
\begin{align}
\mathbf{a}_{\textrm{vps}}(\theta_{k})=\mathbf{a}_{\textrm{v}}(\theta_{k})=\mathbf{I}_{M_{\textrm{v}}}\mathbf{a}_{\textrm{v}}(\theta_{k}),\label{eq:Avhb}
\end{align}
\begin{align}
 & a_{\textrm{hps},m_{\textrm{hr}}}(\theta_{k},\phi_{k})=\left[\mathbf{B}_{\textrm{hps}}^{H}\right]_{m_{\textrm{hr}}+P+1,:}\mathbf{a}_{\textrm{h}}(\theta_{k},\phi_{k})\nonumber \\
 & \approx\sqrt{M_{\textrm{h}}}j^{m_{\textrm{hr}}}J_{m_{\textrm{hr}}}\left(\gamma(\theta_{k})\right)\exp\left(-jm_{\textrm{hr}}\phi_{k}\right).\label{eq:Ahhb}
\end{align}

From \eqref{eq:Avhb} and \eqref{eq:Ahhb}, we see that, through the
proposed $\mathbf{B}_{\textrm{ps}}$, the phases of the array steering
vectors become linear to the element locations. This is important
to exploit the property of the sparse array theory to design the RF-chain
connection matrix $\mathbf{B}_{\textrm{rfc}}$.

\subsection{RF-Chain Connection Network Design}

In this subsection, we apply the sparse array technique to design
$\mathbf{B}_{\textrm{rfc}}$, which enables the DoAs
  of a large number of devices to be estimated with a marginal accuracy loss
while significantly reducing the number of required RF chains. We
aim to use as few RF chains as possible to achieve the same, or even
larger, DoF than the fully connected beamforming array\footnote{Due to the use of phase shifter  network, the antenna DoF of UCyA depends
on the scale of the phase shifter  network. Thus, if we use a fully connected
beamforming array, $(2P+1)M_{\textrm{v}}$ RF chains are needed, which
can provide $\mathcal{O}(PM_{\textrm{v}})$ DoFs. }. This objective is different from the previous sparse array researches,
which have typically focused on maximizing the size of difference
coarrays under the constraint of a fixed number of physical antenna
elements.

We first flatten the 3-D RF-chain connection network of UCyA into
a 2-D plane, as shown in Fig. \ref{fig:3DUCyAto2D}, by disjoining
the RF-chain connection network at the first column phase shifters of
every UCA. Different from typical 2-D arrays, due to the periodicity
of UCAs, the first and the last  phase-shifter output ports of every row in the flattened
2-D RF-chain connection network are identical, as shown in Fig. \ref{fig:3DUCyAto2D},
where the dotted circles denote the last-column phase shifters.
\begin{figure}
\begin{centering}
\includegraphics[width=7cm]{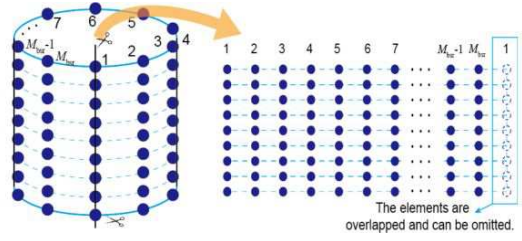}
\par\end{centering}
\centering{}\caption{An illustration of unfolding phase-shifter output ports of a 3-D UCyA to be a 2-D array.\label{fig:3DUCyAto2D}}
\end{figure}

After the flattening processing, the 3-D RF-chain
connection network becomes a quasi-2D rectangular array with size
of $(M_{\textrm{hr}}+1)\times M_{\textrm{vr}}$, where the increased dimension is due to the repeated phase-shifter output ports, as shown in  Fig. \ref{fig:3DUCyAto2D}. The idea of 2-D
sparse arrays can be applied to design a 3-D RF-chain connection network
of UCyA\footnote{Although the RF-chain connection network actually does not have the
exact shape, according to the array steering vectors in \eqref{eq:Avhb}
and \eqref{eq:Ahhb}, we can also regard the network as an UCyA.}. In this paper, we design the RF-chain connection network based on
a 2-D nested array. This is because (1) a nested array can generate
larger hole-free difference coarrays than a coprime array under the
same setting, as discussed in Section III-A; and (2) it has simple
closed-form expressions for a large number of elements, which cannot
be achieved in MRA and MHA. There are also some other frequently-used
2-D sparse arrays, e.g., hourglass arrays and open box arrays (OBAs)
\cite{2DSPARSE,OBA}. We will compare the RF-chain connection networks
designed based on those array geometries with our design in Section V-C.

Our proposed sparse RF-chain connection  network is developed from
the ``Configuration II'' nested array \cite{Nested2D1}. In the
general ``Configuration II'' nested array, when there are $N_{\textrm{dense}}=N_{\textrm{vd}}N_{\textrm{hd}}-1$
and $N_{\textrm{sparse}}=N_{\textrm{vs}}N_{\textrm{hs}}$ elements
in the dense and sparse subarrays, respectively, the constructed hole-free
difference coarray has $N_{\textrm{dc}}=N_{\textrm{vdc}}N_{\textrm{hdc}}=(2N_{\textrm{vd}}N_{\textrm{vs}}-1)N_{\textrm{hd}}N_{\textrm{hs}}$
elements \cite{Nested2D1}, as shown in Fig. \ref{fig:NestedCoprime}(c).
Here,  $N_{\textrm{vd}}$ and $N_{\textrm{hd}}$ are
the numbers of elements in the dense subarray along the vertical and
horizontal directions, respectively; $N_{\textrm{vs}}$ and $N_{\textrm{hs}}$
are the numbers of elements in the sparse subarray along the vertical
and horizontal directions, respectively; and $N_{\textrm{vdc}}$ and
$N_{\textrm{hdc}}$ are the numbers of elements in the difference
coarray along the vertical and horizontal directions, respectively.
We wish to find the distribution of the RF chains between the sparse
and the dense arrays that use as few RF chains as possible to achieve
the same DoF as the fully connected beamforming array.
\begin{figure*}
\begin{centering}
\includegraphics[width=16cm]{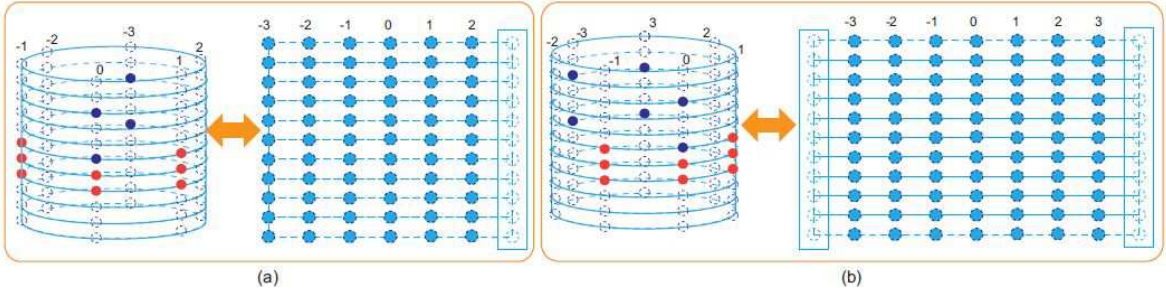}
\par\end{centering}
\centering{}\caption{Two cases of 3-D UCyA unfolding. The locations of the cylindrical
post-phase-shifting ports in sparse and dense subarrays are highlighted
with red and navy blue dots. \label{fig:Two-cases-of}}
\end{figure*}

Due to the above-mentioned periodicity of UCAs, when we apply the
sparse array technique into our hybrid front-end design, two cases
need to be considered for UCyAs. Fig. \ref{fig:Two-cases-of}(a) shows
the first case, where the first and the last columns of the difference
coarray overlap. From Fig. \ref{fig:Two-cases-of}(a), we can see
that because one column (first or last) of the difference coarray
is redundant, two elements of the sparse subarray can be omitted to
reduce the element number. A drawback is that this case requires $M_{\textrm{hr}}=(N_{\textrm{hs}}-1)N_{\textrm{hd}}$,
which would impose a strict requirement on the number
of phase shifters on the horizontal plane. The second case is shown
in Fig. \ref{fig:Two-cases-of}(b), which requires the constructed
difference coarray to be larger than the original UCyA, to achieve
the same DoF as a fully connected beamforming array on the horizontal
space. In this case, $N_{\textrm{hd}}N_{\textrm{hs}}\geq M_{\textrm{hr}}>(N_{\textrm{hs}}-1)N_{\textrm{hd}}/2$.

In our system, due to the new phase shifter  network designed in Section
III-B, we have $M_{\textrm{hr}}=2P+1$ (which is an odd number). Since
in the Configuration II nested array, the dense and sparse subarrays
are symmetric, i.e., both $N_{\textrm{hs}}$ and $N_{\textrm{hd}}$
are odd, we have $(N_{\textrm{hs}}-1)N_{\textrm{hd}}$ is even, and
only the above-mentioned second case needs to be considered in our
system. We formulate the optimization problem as
\begin{align}
\min_{N_{\textrm{vd}},N_{\textrm{hd}},N_{\textrm{vs}},N_{\textrm{hs}}\in\mathbb{Z}^{+}}M_{\textrm{rf}} & =N_{\textrm{vd}}N_{\textrm{hd}}+N_{\textrm{vs}}N_{\textrm{hs}}-1\label{eq:ops}\\
\textrm{s.t. } & \textrm{(C1): }2N_{\textrm{vd}}N_{\textrm{vs}}-1\geq M_{\textrm{vr}},\nonumber \\
 & \textrm{(C2): }N_{\textrm{hd}}N_{\textrm{hs}}\geq M_{\textrm{hr}},\nonumber \\
 & \textrm{(C3): }N_{\textrm{vd}},N_{\textrm{hd}}>1,\nonumber \\
 & \textrm{(C4): }N_{\textrm{hd}}\,\textrm{is odd and }N_{\textrm{vd}}/N_{\textrm{hd}}\in\mathbb{Z}^{+}.\nonumber
\end{align}
(C1) and (C2) guarantee the DoF requirements for the constructed difference
coarray of the RF-chain connection network. (C3) avoids solutions
that degenerate to 1-D arrays. (C4) is due to the fact that the dense
array in Configuration II is symmetrical, and $N_{\textrm{hd}}$ and
$N_{\textrm{vd}}$ are invariant factors \cite{Nested2D2} of the
array distribution matrix.

The solution for \eqref{eq:ops} can be obtained by adopting the following
strategy. According to (C1) and (C2), since $N_{\textrm{vs}},N_{\textrm{hs}}\in\mathbb{Z}^{+}$,
we can obtain $N_{\textrm{vs}}=\left\lceil (M_{\textrm{vr}}-1)/2N_{\textrm{vd}}\right\rceil $
and $N_{\textrm{hs}}=\left\lceil M_{\textrm{hr}}/N_{\textrm{hd}}\right\rceil $.
The optimization problem \eqref{eq:ops} becomes
\begin{align}
\min_{N_{\textrm{vd}},N_{\textrm{hd}}\in\mathbb{Z}^{+}}M_{\textrm{rf}} & =N_{\textrm{vd}}N_{\textrm{hd}}+\left\lceil \frac{(M_{\textrm{vr}}-1)/2}{N_{\textrm{vd}}}\right\rceil \left\lceil \frac{M_{\textrm{hr}}}{N_{\textrm{hd}}}\right\rceil \label{eq:ops-1}\\
\textrm{s.t. } & \textrm{(C3) and (C4)}\nonumber
\end{align}
Given $M_{\textrm{vr}}$ and $M_{\textrm{hr}}$, we see that \eqref{eq:ops-1}
exhibits the form of $y=x+\frac{a}{x}$, where $a>0$ is a constant
and $y=x+\frac{a}{x}\geq2\sqrt{a}$. Because $y=2\sqrt{a}$
iff $x=\frac{a}{x}$, the minimum $M_{\textrm{rf}}$ can be obtained
when the difference between $N_{\textrm{vd}}N_{\textrm{hd}}$ and
$\left\lceil (M_{\textrm{vr}}-1)/2N_{\textrm{vd}}\right\rceil \left\lceil M_{\textrm{hr}}/N_{\textrm{hd}}\right\rceil $
is the smallest. Since $N_{\textrm{vd}},N_{\textrm{hd}},N_{\textrm{vs}},N_{\textrm{hs}}\in\mathbb{Z}^{+}$,
we can determine the approximate value ranges of $N_{\textrm{vd}}N_{\textrm{hd}}$
and $\left\lceil (M_{\textrm{vr}}-1)/2N_{\textrm{vd}}\right\rceil \left\lceil M_{\textrm{hr}}/N_{\textrm{hd}}\right\rceil $,
and \eqref{eq:ops-1} is an integer programming problem. According
to (C3) and (C4), the optimal solutions of $N_{\textrm{vd}},N_{\textrm{hd}},N_{\textrm{vs}},$
and $N_{\textrm{hs}}$ to \eqref{eq:ops-1} can be obtained by using
brute-force search with the value range between $N_{\textrm{vd}}N_{\textrm{hd}}$
and $\left\lceil (M_{\textrm{vr}}-1)/2N_{\textrm{vd}}\right\rceil \left\lceil M_{\textrm{hr}}/N_{\textrm{hd}}\right\rceil $.

In the proposed sparse RF-chain connection network, the RF chains
only need to connect the phase shifters located in the dense and sparse
subarrays. Based on the calculated values of $N_{\textrm{vd}},N_{\textrm{hd}},N_{\textrm{vs}},$
and $N_{\textrm{hs}}$, now we provide the element locations in the
dense and sparse arrays. For illustration convenience, we define the
overlapping point of the sparse and dense arrays as the origin of
the nested array\footnote{Because the parameter estimation depends on the difference between
array elements, the changed absolute positions of array elements does
not effect the estimation performance.}. Let $\mathbf{m}_{\textrm{sp}}=(m_{\textrm{v\_sp}},m_{\textrm{h\_sp}})$
and $\mathbf{m}_{\textrm{de}}=(m_{\textrm{v\_de}},m_{\textrm{h\_de}})$
as the locations of elements in the dense and sparse arrays (to which
the RF chains connect), respectively. We have
\begin{equation}
\begin{cases}
m_{\textrm{v\_sp}}=N_{\textrm{vd}}(n_{\textrm{vs}}-1),\\
m_{\textrm{h\_sp}}=N_{\textrm{hd}}(-N_{\textrm{hs}}/2+n_{\textrm{hs}}-1/2),
\end{cases}\label{eq:sparsearray}
\end{equation}
\begin{equation}
\begin{cases}
m_{\textrm{v\_de}}=-N_{\textrm{vd}}+n_{\textrm{vd}},\\
m_{\textrm{h\_de}}=-(N_{\textrm{hd}}-1)/2+n_{\textrm{hd}}-1,
\end{cases}\label{eq:densearray}
\end{equation}
where $n_{\textrm{vs}}=1,2,\ldots,N_{\textrm{vs}}$; $n_{\textrm{hs}}=1,2,\ldots,N_{\textrm{hs}}$;
$n_{\textrm{vd}}=1,2,\ldots,N_{\textrm{vd}}$; and $n_{\textrm{hd}}=1,2,\ldots,N_{\textrm{hd}}$.
Let $\mathbb{M}_{\textrm{rf\_d}}=\left\{ \mathbf{m}_{\textrm{v\_de}}\otimes\mathbf{m}_{\textrm{h\_de}}\right\} $
and $\mathbb{M}_{\textrm{rf\_s}}=\left\{ \mathbf{m}_{\textrm{v\_sp}}\otimes\mathbf{m}_{\textrm{h\_sp}}\right\} $
denote the sets of the RF-chain connection points in the dense and
sparse arrays, respectively. $\mathbf{m}_{\textrm{v\_de}}\in\mathbb{C}^{N_{\textrm{vd}}\times1}$
and $\mathbf{m}_{\textrm{h\_de}}\in\mathbb{C}^{N_{\textrm{hd}}\times1}$
are the element locations of the dense array along the vertical and
horizontal directions, respectively. $\mathbf{m}_{\textrm{v\_sp}}\in\mathbb{C}^{N_{\textrm{vs}}\times1}$
and $\mathbf{m}_{\textrm{h\_sp}}\in\mathbb{C}^{N_{\textrm{hs}}\times1}$
are the element locations of the sparse array along the vertical and
horizontal directions, respectively. The set of RF-chain connection
points is $\mathbb{M}_{\textrm{rfc}}=\left\{ \mathbb{M}_{\textrm{rf\_d}}\cup\mathbb{M}_{\textrm{rf\_s}}\right\} .$
The constructed RF-chain connection matrix $\mathbf{B}_{\textrm{rfc}}$
is given by
\begin{align}
 & \left[\mathbf{B}_{\textrm{rfc}}\right]_{m_{\textrm{bsr}},m_{\textrm{rf}}}=\nonumber \\
 & \begin{cases}
1, & \textrm{if }m_{\textrm{bsr}}\in\mathbb{M}_{\textrm{rfc}}\textrm{ and }\\
 & \left[\mathbf{B}_{\textrm{rfc}}\right]_{m'_{\textrm{bsr}}\neq m_{\textrm{bsr}},m_{\textrm{rf}}}=\left[\mathbf{B}_{\textrm{rfc}}\right]_{m_{\textrm{bsr}},m'_{\textrm{rf}}\neq m_{\textrm{rf}}}=0;\\
0, & \textrm{otherwise.}
\end{cases}
\end{align}

By deploying the proposed sparse RF-chain connection network and using
the second-order statistics of the received signal for channel estimation,
the DoF of the proposed network is $O((2N_{\textrm{vd}}N_{\textrm{vs}}-1)\times M_{\textrm{hr}})$.
In other words, according to \cite{Nested2D2}, we can estimate the
channel parameters of $(N_{\textrm{vdc}}-1)(N_{\textrm{hdc}}-1)$
devices by using the proposed hybrid front-end. Because there is an
overlapping point at the origin, the total number of RF chains required
in our network is $M_{\textrm{rf}}=N_{\textrm{vd}}N_{\textrm{hd}}+N_{\textrm{vs}}N_{\textrm{hs}}-1$,
as shown in \eqref{eq:ops}\footnote{It can also be proved that the phase shifter at this location is useless
and does not need to be connected \cite{Nested2D1}. However, for
ease of description, here we assume that this phase shifter is connected
in our network, which does not affect the results in this paper.}. Due to the periodicity of UCAs, there are only
up to $N_{\textrm{hdc}}=M_{\textrm{hr}}$ virtual elements on the
horizontal plane of the constructed 3-D difference coarray of the
RF-chain connection network. Along the vertical direction, the number
of virtual elements is $N_{\textrm{vdc}}=2N_{\textrm{vd}}N_{\textrm{vs}}-1$.

According to the element locations of the dense and sparse arrays
in \eqref{eq:sparsearray} and \eqref{eq:densearray}, we also provide
the element locations of the constructed difference coarray. Let $\mathbf{m}_{\textrm{dc}}=(m_{\textrm{v\_dc}},m_{\textrm{h\_dc}})$,
where $m_{\textrm{v\_dc}}$ and $m_{\textrm{h\_dc}}$ correspond to
the locations along the vertical and horizontal directions, respectively.
We have
\begin{align}
m_{\textrm{v\_dc}} & =-(N_{\textrm{vdc}}-1)/2+n_{\textrm{vdc}}-1=-N_{\textrm{vd}}N_{\textrm{vs}}+n_{\textrm{vdc}},
\end{align}
\begin{align}
m_{\textrm{h\_dc}} & =-(N_{\textrm{hdc}}-1)/2+n_{\textrm{hdc}}-1=-P+n_{\textrm{hdc}}-1,
\end{align}
where $n_{\textrm{vdc}}=1,2,\ldots,N_{\textrm{vdc}}$ and $n_{\textrm{hdc}}=1,2,\ldots,N_{\textrm{hdc}}$.
The shape of the constructed 3-D difference coarray RF-chain connection
network is the same as the UCyA, but the former has a larger DoF.

The signals through the proposed RF-chain connection network are given
by
\begin{equation}
\mathbf{x}_{\textrm{sn},m_{\textrm{t}}}=\sum_{^{k=1}}^{K}s_{m_{\textrm{t}},k}\mathbf{a}_{\textrm{sn}}(\phi_{k},\theta_{k})+\mathbf{n}_{\textrm{sn},m_{\textrm{t}}},\label{eq:signals_byPSN}
\end{equation}
where
\begin{align*}
 & \mathbf{a}_{\textrm{sn}}(\phi_{k},\theta_{k})\\
 & =\left[\begin{array}{c}
\mathbf{a}_{\textrm{sn,s}}(\phi_{k},\theta_{k})\\
\mathbf{a}_{\textrm{sn,d}}(\phi_{k},\theta_{k})
\end{array}\right]=\left[\begin{array}{c}
\mathbf{a}_{\textrm{\textrm{sn,s}v}}(\theta_{k})\otimes\mathbf{a}_{\textrm{\textrm{sn,s}h}}(\theta_{k},\phi_{k})\\
\mathbf{a}_{\textrm{\textrm{sn,d}v}}(\theta_{k})\otimes\mathbf{a}_{\textrm{\textrm{sn,d}h}}(\theta_{k},\phi_{k})
\end{array}\right].
\end{align*}
The elements of $\mathbf{a}_{\textrm{\textrm{sn,s}v}}(\theta_{k})\in\mathbb{C}^{N_{\textrm{vs}}\times1}$
and $\mathbf{a}_{\textrm{\textrm{sn,s}h}}(\theta_{k},\phi_{k})\in$ $\mathbb{C}^{N_{\textrm{hs}}\times1}$
are $a_{\textrm{\textrm{sn,s}v},n_{\textrm{vs}}}(\theta_{k})=a_{\textrm{vs,}m_{\textrm{v\_sp}}}(\theta_{k})$
and $a_{\textrm{\textrm{sn,s}h},n_{\textrm{hs}}}(\theta_{k},\phi_{k})=$ $a_{\textrm{hs},m_{\textrm{h\_sp}}}(\theta_{k},\phi_{k})$,
respectively, where $a_{\textrm{vs,}m_{\textrm{v\_sp}}}(\theta_{k})$
and $a_{\textrm{hs},m_{\textrm{h\_sp}}}(\theta_{k},\phi_{k})$ are
the array steering vectors of the sparse subarray along the vertical
and horizontal directions, respectively. The elements of the array
steering vectors of the dense subarray, i.e., $\mathbf{a}_{\textrm{\textrm{sn,d}v}}(\theta_{k})\in\mathbb{C}^{N_{\textrm{vd}}\times1}$
and $\mathbf{a}_{\textrm{\textrm{sn,d}h}}(\theta_{k},\phi_{k})\in\mathbb{C}^{N_{\textrm{hd}}\times1}$,
can be written in the same way. Here, $\mathbf{n}_{\textrm{sn},m_{\textrm{t}}}\in\mathbb{C}^{M_{\textrm{rf}}\times1}$
is the noise component through the RF-chain connection network.

The signal model \eqref{eq:signals_byPSN} can also be rewritten as
\begin{equation}
\mathbf{x}_{\textrm{sn},m_{\textrm{t}}}=\mathbf{A}_{\textrm{sn}}\mathbf{s}_{m_{\textrm{t}}}+\mathbf{n}_{\textrm{sn},m_{\textrm{t}}},\label{eq:signals_byPSN-1}
\end{equation}
where $\mathbf{A}_{\textrm{sn}}=[\mathbf{a}_{\textrm{sn}}(\phi_{1},\theta_{1}),\mathbf{a}_{\textrm{sn}}(\phi_{2},\theta_{2}),\ldots,\mathbf{a}_{\textrm{sn}}(\phi_{K},\theta_{K})]\in\mathbb{C}^{M_{\textrm{rf}}\times K}$
and $\mathbf{s}_{m_{\textrm{t}}}=[s_{m_{\textrm{t}},1},s_{m_{\textrm{t}},2},\ldots,s_{m_{\textrm{t}},K}]^{T}\in\mathbb{C}^{K\times1}$.

By calculating the autocorrelation of $\mathbf{x}_{\textrm{sn},m_{\textrm{t}}}$,
we have
\begin{equation}
\mathbf{R}_{\textrm{sn,}m_{\textrm{t}}}=\textrm{E}\left\{ \mathbf{x}_{\textrm{sn},m_{\textrm{t}}}\mathbf{x}_{\textrm{sn},m_{\textrm{t}}}^{H}\right\} =\mathbf{A}_{\textrm{sn}}\mathbf{R}_{\textrm{ss},m_{\textrm{t}}}\mathbf{A}_{\textrm{sn}}^{H}+\mathbf{R}_{\textrm{nn},m_{\textrm{t}}},\label{eq:Calculatingautocorrelation}
\end{equation}
where $\mathbf{R}_{\textrm{ss},m_{\textrm{t}}}=\textrm{diag}\left(\mathbf{\sigma}_{\textrm{s},m_{\textrm{t}},1}^{2},\ldots,\mathbf{\sigma}_{\textrm{s},m_{\textrm{t}},K}^{2}\right)$
and $\mathbf{R}_{\textrm{nn},m_{\textrm{t}}}=\textrm{diag}\left(\mathbf{\sigma}_{\textrm{n},m_{\textrm{t}},1}^{2},\ldots,\mathbf{\sigma}_{\textrm{n},m_{\textrm{t}},K}^{2}\right)$
are the autocorrelation matrices of $\mathbf{s}_{m_{\textrm{t}}}$
and $\mathbf{n}_{\textrm{sn},m_{\textrm{t}}}$, respectively.

We vectorize $\mathbf{R}_{\textrm{sn,}m_{\textrm{t}}}$ as
\begin{equation}
\mathbf{\mathbf{y}_{\textrm{vR\textrm{,}\ensuremath{m_{\textrm{t}}}}}}=\textrm{vec}(\mathbf{R}_{\textrm{sn,}m_{\textrm{t}}})=\left[\mathbf{A}_{\textrm{sn}}^{*}\diamond\mathbf{A}_{\textrm{sn}}\right]\mathbf{d}_{m_{\textrm{t}}}+\textrm{vec}(\mathbf{R}_{\textrm{nn},m_{\textrm{t}}}),\label{eq:VEC}
\end{equation}
where $\left[\mathbf{d}_{m_{\textrm{t}}}\right]_{k,1}=\mathbf{\sigma}_{\textrm{s},m_{\textrm{t}},k}^{2}$,
and $\mathbf{\sigma}_{\textrm{s},m_{\textrm{t}},k}^{2}$ is the power
of the $k$-th signal. The $k$-th column of the
matrix $\left[\mathbf{A}_{\textrm{sn}}^{*}\diamond\mathbf{A}_{\textrm{sn}}\right]$
contains elements representing the cross-differences between sparse
and dense subarrays,  i.e., $a_{\textrm{sn,s},\mathbf{m}_{\textrm{sp}}}^{*}(\phi_{k},\theta_{k})a_{\textrm{sn,d},\mathbf{m}_{\textrm{de}}}(\phi_{k},\theta_{k})$
and $a_{\textrm{sn,s},\mathbf{m}_{\textrm{de}}}^{*}(\phi_{k},\theta_{k})a_{\textrm{sn,d},\mathbf{m}_{\textrm{sp}}}(\phi_{k},\theta_{k}),$
and the self-differences of sparse and dense subarrays, i.e., $a_{\textrm{sn,s},\mathbf{m}_{\textrm{sp},1}}^{*}(\phi_{k},\theta_{k})a_{\textrm{sn,s},\mathbf{m}_{\textrm{sp},2}}(\phi_{k},\theta_{k})$
and $a_{\textrm{sn,d},\mathbf{m}_{\textrm{de,1}}}^{*}(\phi_{k},\theta_{k})a_{\textrm{sn,d},\mathbf{m}_{\textrm{de,2}}}(\phi_{k},\theta_{k})$.
Here, $\mathbf{m}_{\textrm{sp},1}$ and $\mathbf{m}_{\textrm{sp},2}$
denote that $a_{\textrm{sn,s},\mathbf{m}_{\textrm{sp},1}}(\phi_{k},\theta_{k})$
and $a_{\textrm{sn,s},\mathbf{m}_{\textrm{sp},2}}(\phi_{k},\theta_{k})$
are different elements in the sparse subarray, and $\mathbf{m}_{\textrm{de,1}}$
and $\mathbf{m}_{\textrm{de,2}}$ denote that $a_{\textrm{sn,d},\mathbf{m}_{\textrm{de,1}}}(\phi_{k},\theta_{k})$
and $a_{\textrm{sn,d},\mathbf{m}_{\textrm{de,2}}}(\phi_{k},\theta_{k})$
are different elements in the dense subarray.

We sort the rows of $\mathbf{\mathbf{y}_{\textrm{vR\textrm{,}\ensuremath{m_{\textrm{t}}}}}}$
in the ascending order of their phases, and then remove the redundant
rows with the same phases. Then, we can obtain the array steering
vector of the difference coarray $\mathbf{A}_{\textrm{df}}\in\mathbb{C}^{N_{\textrm{vdc}}N_{\textrm{hdc}}\times1}$
from $\left[\mathbf{A}_{\textrm{sn}}^{*}\diamond\mathbf{A}_{\textrm{sn}}\right]$.
We also calculate and store the mean of the ``nonzero'' rows of
$\textrm{vec}(\mathbf{R}_{\textrm{nn},m_{\textrm{t}}})$\footnote{All the rows of $\textrm{vec}(\mathbf{R}_{\textrm{nn},m_{\textrm{t}}})$
with nonzero value correspond to the phase difference of 0 in the
different coarray, which are produced by the self difference of sparse
and dense subarrays. Because the noise is temporally and spatially
white with power $\mathbf{\sigma}_{\textrm{n}}^{2}$, by averaging
the value of these rows, we have $\mathbf{\sigma}_{\textrm{n}}^{2}=\sum_{k=1}^{K}\mathbf{\sigma}_{\textrm{n},m_{\textrm{t}},k}^{2}$.}, and obtain
\begin{equation}
\mathbf{\mathbf{y}_{\textrm{df\textrm{,}\ensuremath{m_{\textrm{t}}}}}}=\mathbf{A}_{\textrm{df}}\mathbf{d}_{m_{\textrm{t}}}+\mathbf{\sigma}_{\textrm{n}}^{2}\mathbf{e}_{\textrm{df}},\label{eq:dfmodel}
\end{equation}
where $\left[\mathbf{A}_{\textrm{df}}\right]_{:,k}=\mathbf{a}_{\textrm{df}}(\phi_{k},\theta_{k})=\mathbf{a}_{\textrm{\textrm{df}v}}(\theta_{k})\otimes\mathbf{a}_{\textrm{\textrm{df}h}}(\theta_{k},\phi_{k})$
and $\mathbf{e}_{\textrm{df}}\in\mathbb{C}^{N_{\textrm{vdc}}N_{\textrm{hdc}}\times1}$
is a vector of all zeros except a ``1'' at the $(P+1)N_{\textrm{vd}}N_{\textrm{vs}}$-th
entry. The element phases of $\mathbf{a}_{\textrm{\textrm{df}v}}(\theta_{k})\in\mathbb{C}^{N_{\textrm{vdc}}\times1}$
and $\mathbf{a}_{\textrm{\textrm{df}h}}(\theta_{k},\phi_{k})\in\mathbb{C}^{N_{\textrm{hdc}}\times1}$
are given by
\begin{equation}
a_{\textrm{\textrm{df}v},n_{\textrm{vdc}}}(\theta_{k})=\frac{1}{M_{\textrm{v}}}\exp\left(-j\frac{2\pi}{\lambda}hm_{\textrm{v\_dc}}\cos(\theta_{k})\right),
\end{equation}
\begin{equation}
a_{\textrm{\textrm{df}h},n_{\textrm{hdc}}}(\theta_{k},\phi_{k})=\xi_{m_{\textrm{h\_dc}}}\left(\theta_{k}\right)\exp\left(-jm_{\textrm{h\_dc}}\phi_{k}\right),
\end{equation}
where $\xi_{m_{\textrm{h\_dc}}}\left(\theta_{k}\right)=M_{\textrm{h}}j^{m_{\textrm{h\_dc}}}J_{m_{\textrm{h\_ds},1}}\left(\gamma(\theta_{k})\right)J_{m_{\textrm{h\_ds},2}}\left(\gamma(\theta_{k})\right),$
and $m_{\textrm{h\_dc}}=m_{\textrm{h\_ds},1}-m_{\textrm{h\_ds},2}$
($m_{\textrm{h\_ds},1},m_{\textrm{h\_ds},2}\in\mathcal{M}_{\textrm{h\_ds}})$.
$\mathcal{M}_{\textrm{h\_ds}}=\left\{ m_{\textrm{h\_de}},m_{\textrm{h\_sp}}\right\} $
collects the horizontal locations of the elements in the dense and
sparse arrays. \eqref{eq:dfmodel} can be viewed as the signal $\mathbf{d}_{m_{\textrm{t}}}$
received at an array with steering matrix $\mathbf{A}_{\textrm{df}}$.

Now we formulate the received samples in the tensor form. We first
decompose $\mathbf{\mathbf{y}_{\textrm{df\textrm{,}\ensuremath{m_{\textrm{t}}}}}}$
into the vertical and horizontal domains (corresponding to the first
and second modes of the tensor model), as given by $\mathbf{\mathbf{Y}_{\textrm{df\textrm{,}\ensuremath{m_{\textrm{t}}}}}}=\textrm{invec}(\mathbf{\mathbf{y}_{\textrm{df\textrm{,}\ensuremath{m_{\textrm{t}}}}}})\in\mathbb{C}^{(2N_{\textrm{vd}}N_{\textrm{vs}}-1)\times M_{\textrm{hr}}}$.
Then, we collect $\mathbf{Y}_{\textrm{df\textrm{,}\ensuremath{m_{\textrm{t}}}}}$
at all time frames, and store them in the time domain (corresponding
to the third mode of the tensor model). Thus, the
received samples can be expressed as
\begin{align}
\mathcal{Y}_{\textrm{df}} & =[\mathbf{Y}_{\textrm{df\textrm{,}\ensuremath{1}}}\sqcup_{3}\mathbf{Y}_{\textrm{df\textrm{,}\ensuremath{2}}}\sqcup_{3}\ldots\sqcup_{3}\mathbf{Y}_{\textrm{df\textrm{,}}M_{\textrm{t}}}]\nonumber \\
 & =\mathcal{A}_{\textrm{df}}\times_{3}\mathbf{D}+\mathcal{N}_{\textrm{df}}\in\mathbb{C}^{N_{\textrm{vdc}}\times N_{\textrm{hdc}}\times M_{\textrm{t}}},\label{eq:modeltensorfinal}
\end{align}
where $\mathbf{D}=[\mathbf{d}_{1},\mathbf{d}_{2},\ldots,\mathbf{d}_{M_{\textrm{t}}}]^{T}\in\mathbb{C}^{M_{\textrm{t}}\times K}$,
$\mathcal{A}_{\textrm{df}}\in\mathbb{C}^{N_{\textrm{vdc}}\times N_{\textrm{hdc}}\times K}$
is known as the space-time response tensor \cite{SCIVanderveen},
and $\mathcal{N}_{\textrm{df}}$ is the noise tensor model. Due to
the above-mentioned process \eqref{eq:dfmodel}-\eqref{eq:modeltensorfinal},
the elements of $\mathcal{N}_{\textrm{df}}$ are all zeros except
$\mathbf{\sigma}_{\textrm{n}}^{2}$ at $(0,0,m_{\textrm{t}}),$ $m_{\textrm{t}}=1,2,\ldots,M_{\textrm{t}}$.
In \eqref{eq:modeltensorfinal}, $\mathcal{A}_{\textrm{df}}$ is obtained
as
\begin{align}
 & \mathcal{A}_{\textrm{df}}=[\mathbf{a}_{\textrm{\textrm{df}v}}(\theta_{1})\circ\mathbf{a}_{\textrm{\textrm{df}h}}(\theta_{1},\phi_{1})\sqcup_{3}\mathbf{a}_{\textrm{\textrm{df}v}}(\theta_{2})\circ\mathbf{a}_{\textrm{\textrm{df}h}}(\theta_{2},\phi_{2})\nonumber \\
 & \quad\sqcup_{3}\ldots\sqcup_{3}\mathbf{a}_{\textrm{\textrm{df}v}}(\theta_{K})\circ\mathbf{a}_{\textrm{\textrm{df}h}}(\theta_{K},\phi_{K})].\label{eq:STresponsetensor}
\end{align}

By substituting \eqref{eq:STresponsetensor} into \eqref{eq:modeltensorfinal},
we obtain
\begin{align}
\mathcal{Y}_{\textrm{df}} & =\sum_{^{k=1}}^{K}\mathbf{a}_{\textrm{\textrm{df}v}}(\theta_{k})\circ\mathbf{a}_{\textrm{\textrm{df}h}}(\theta_{k},\phi_{k})\circ\left[\mathbf{D}\right]_{:,k}+\mathcal{N}_{\textrm{df}}\nonumber \\
 & =\left\llbracket \mathcal{Z}_{\textrm{df}};\mathbf{A}_{\textrm{\textrm{df}v}},\mathbf{A}_{\textrm{\textrm{df}h}},\mathbf{D}\right\rrbracket +\mathcal{N}_{\textrm{df}},\label{eq:tensorDDD}
\end{align}
where $\left[\mathbf{A}_{\textrm{\textrm{df}v}}\right]_{:,k}=\mathbf{a}_{\textrm{\textrm{df}v}}(\theta_{k}),$
$\left[\mathbf{A}_{\textrm{\textrm{df}h}}\right]_{:,k}=\mathbf{a}_{\textrm{\textrm{df}h}}(\theta_{k},\phi_{k})$,
and $\mathcal{Z}_{\textrm{df}}\in\mathbb{C}^{K\times K\times K}$
is an order-3 identity superdiagonal tensor\footnote{A tensor $\mathcal{A}\in\mathbb{C}^{I_{1}\times I_{2}\times\cdots\times I_{N}}$
is diagonal if $a_{i_{1}i_{2}\cdots i_{N}}\neq0$ only if $i_{1}=i_{2}=\cdots=i_{N}$.
When $I_{1}=I_{2}=\cdots=I_{N}$, $\mathcal{A}$ is called as superdiagonal.}.

Eq. \eqref{eq:modeltensorfinal} shows that the elements of the equivalent
signal matrix $\mathbf{D}\in\mathbb{C}^{M_{\textrm{t}}\times K}$
are actually the received signal powers due to the autocorrelation
calculation \eqref{eq:Calculatingautocorrelation}. To build a full-rank
matrix $\mathbf{D}$ for DoA estimation, one would need to assume
that the received signal powers change over time, and the power of
every signal is different from each other, as assumed in \cite{NestedNoSmoothing}.
However, such assumption is unrealistic in practice. It
is possible that the rank of the equivalent device signal matrix $\mathbf{D}$
is smaller than the number of devices $K$, i.e., $\textrm{Rank}(\mathbf{D})<K$,
which behaves as if some of the received signals are coherent, leading
to incorrect channel estimation. Prevent possible coherent signals,
we propose a novel approach in the next section to construct a signal
tensor model with suitable $n$-ranks in all modes. This allows us
to estimate the 2-D DoAs of $K$ devices.

\section{Spatial smoothing-based tensor $n$-rank enhancement}

In this section, we analyze the relationship between the rank of $\mathbf{D}$
and the $n$-rank of $\mathcal{Y}_{\textrm{df}}$. We propose a spatial
smoothing-based method to enhance the $n$-rank of $\mathcal{Y}_{\textrm{df}}$.
By using the proposed method, we verify that one can build a signal
tensor model that provides a large enough rank in each mode to perform
the DoA estimation of $K$ devices, even when the received signal
powers of all the devices are equal. These powers are steady temporally
across all time frames.

As discussed in Section III-C, the rank of $\mathbf{D}$
in \eqref{eq:tensorDDD} is typically smaller than the number of devices,
$K$, in practice. Based on the uniqueness condition of tensor CP
decomposition \cite{TensorDecompositions2}, we first provide the
following theorem to evaluate the impact of $\textrm{Rank}(\mathbf{D})$
on the $n$-ranks of the tensor $\mathcal{Y}_{\textrm{df}}$.

\noindent \begin{theorem} For $\mathcal{Y}_{\textrm{df}}=\left\llbracket \mathcal{Z}_{\textrm{df}};\mathbf{A}_{\textrm{\textrm{df}v}},\mathbf{A}_{\textrm{\textrm{df}h}},\mathbf{D}\right\rrbracket +\mathcal{N}_{\textrm{df}},$
if $\textrm{Rank}(\mathbf{D})<K$, the ranks of the signal spaces
of $\mathcal{Y}_{\textrm{df}}$ in all modes are smaller than the
number of devices $K$, i.e., $\textrm{Rank}(\mathbf{U}_{\textrm{v,}n})<K,$
$n=1,2,3$, where $\mathbf{U}_{\textrm{v,}n}$ is the mode-$n$ signal
subspace of $\mathcal{Y}_{\textrm{df}}$ with $\mathbf{U}_{\textrm{v,}1}\in\mathbb{C}^{N_{\textrm{vdc}}\times K}$,
$\mathbf{U}_{\textrm{v,}2}\in\mathbb{C}^{N_{\textrm{hdc}}\times K}$,
and $\mathbf{U}_{\textrm{v,}3}\in\mathbb{C}^{M_{\textrm{t}}\times K}$.

\noindent \end{theorem} \begin{proof} See Appendix
II.\end{proof}

According to Theorem 2, when the rank of $\mathbf{D}$
in \eqref{eq:tensorDDD} is smaller than the number of devices $K$,
we cannot decompose the tensor model \eqref{eq:modeltensorfinal}
into the signal and noise spaces in all modes. As a result, the subspace-based
algorithms cannot be used to estimate the angles of the devices. To
enhance the $n$-rank of the signal tensor model, we apply spatial
smoothing techniques \cite{ICCMearly} to build
up a sample tensor model whose signal subspace is full rank in each
mode.
\begin{figure*}
\begin{centering}
\includegraphics[width=16cm]{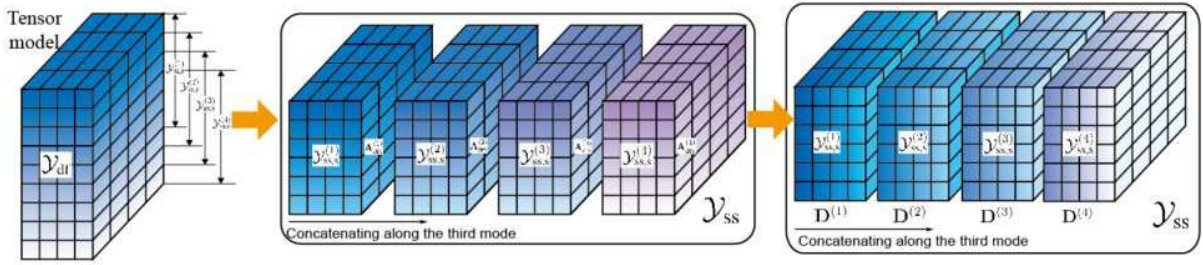}
\par\end{centering}
\centering{}\caption{An illustration of the proposed tensor $n$-rank enhancement method.
\label{fig:SpatialSmoothingTensorRank}}
\end{figure*}

We divide $\mathcal{Y}_{\textrm{df}}$ in \eqref{eq:modeltensorfinal}
into $N_{\textrm{is}}$ identical subtensors in its first mode, as
shown in the left-hand side of Fig. \ref{fig:SpatialSmoothingTensorRank}.
The $n_{\textrm{is}}$-th subtensor $(n_{\textrm{is}}=1,2,\ldots,N_{\textrm{is}})$
can be constructed as
\begin{equation}
\mathcal{Y}_{\textrm{ss}}^{(n_{\textrm{is}})}=\mathcal{Y}_{\textrm{df}}\times_{1}\mathbf{J}_{\textrm{ss,}n_{\textrm{is}}}\in\mathbb{C}^{N_{\textrm{ss}}\times N_{\textrm{hdc}}\times M_{\textrm{t}}},\label{eq:subtensors}
\end{equation}
 where $\mathbf{J}_{\textrm{ss,}n_{\textrm{is}}}=[\mathbf{0}_{N_{\textrm{ss}}\times(n_{\textrm{is}}-1)},\mathbf{I}_{N_{\textrm{ss}}},\mathbf{0}_{N_{\textrm{ss}}\times(N_{\textrm{is}}-n_{\textrm{is}})}]$
and $N_{\textrm{ss}}=2N_{\textrm{vd}}N_{\textrm{vs}}-N_{\textrm{is}}$.

We can see that
\begin{equation}
\mathcal{Y}_{\textrm{ss}}^{(n_{\textrm{is}})}=\left\llbracket \mathcal{Z}_{\textrm{df}};\mathbf{A}_{\textrm{\textrm{df}v}}^{(n_{\textrm{is}})},\mathbf{A}_{\textrm{\textrm{df}h}},\mathbf{D}\right\rrbracket +\mathcal{N}^{(n_{\textrm{is}})},
\end{equation}
where $\mathbf{A}_{\textrm{\textrm{df}v}}^{(n_{\textrm{is}})}=\mathbf{J}_{\textrm{ss,}n_{\textrm{is}}}\mathbf{A}_{\textrm{\textrm{df}v}}=\mathbf{A}_{\textrm{\textrm{df}v}}^{(1)}\mathbf{Q}_{\textrm{ss}}^{n_{\textrm{is}}-1}\in\mathbb{C}^{N_{\textrm{ss}}\times K}$,
$\mathbf{Q}_{\textrm{ss}}=\textrm{diag}\left(q_{\textrm{ss},1},q_{\textrm{ss},2},\ldots,q_{\textrm{ss},K}\right)\in\mathbb{C}^{K\times K},$
$q_{\textrm{ss},k}=e^{j\frac{2\pi}{\lambda}h\cos(\theta_{k})}$, and
$\mathcal{N}^{(n_{\textrm{is}})}=\mathcal{N}_{\textrm{df}}\times_{1}\mathbf{J}_{\textrm{ss,}n_{\textrm{is}}}\in\mathbb{C}^{N_{\textrm{ss}}\times N_{\textrm{hdc}}\times M_{\textrm{t}}}$
is the selected subtensor of the noise model. We can verify that
only when $N_{\textrm{vd}}N_{\textrm{vs}}-N_{\textrm{ss}}+1\leq n_{\textrm{is}}\leq N_{\textrm{vd}}N_{\textrm{vs}}$,
$\mathcal{N}^{(n_{\textrm{is}})}$ has $\mathbf{\sigma}_{\textrm{n}}^{2}$
at the $n_{\textrm{ss}}=(N_{\textrm{vd}}N_{\textrm{vs}}-n_{\textrm{is}}+1)$-th
entry of the first mode, while $n_{\textrm{hdc}}=P+1$ and $m_{\textrm{t}}=1,2,\ldots,M_{\textrm{t}}$.
In all other cases, $\mathcal{N}^{(n_{\textrm{is}})}=0$.

By concatenating the $N_{\textrm{is}}$ identical subtensors $\mathcal{Y}_{\textrm{ss}}^{(n_{\textrm{is}})}$,
$n_{\textrm{is}}=1,\ldots,N_{\textrm{is}}$, as shown in the middle
block of Fig. \ref{fig:SpatialSmoothingTensorRank}, the spatially
smoothed signal tensor model can be constructed as
\begin{align}
\mathcal{Y}_{\textrm{ss}} & =\left[\underset{{\scriptstyle n_{\textrm{is}}=1,\ldots,N_{\textrm{is}}}}{\sqcup_{3}}\mathcal{Y}_{\textrm{ss}}^{(n_{\textrm{is}})}\right]\in\mathbb{C}^{N_{\textrm{ss}}\times N_{\textrm{hdc}}\times\left(M_{\textrm{t}}N_{\textrm{is}}\right)},\label{eq:smoothingModel}
\end{align}
which has a rank large enough in each mode to perform DoA estimation
of the $K$ devices.

Now, we proceed to verify the $n$-ranks of $\mathcal{Y}_{\textrm{ss}}$.
Define $\mathcal{Y}_{\textrm{ss,s}}^{(n_{\textrm{is}})}=\left\llbracket \mathcal{Z}_{\textrm{df}};\mathbf{A}_{\textrm{\textrm{df}v}}^{(n_{\textrm{is}})},\mathbf{A}_{\textrm{\textrm{df}h}},\mathbf{D}\right\rrbracket .$
We have
\begin{align}
 & \left[\mathcal{Y}_{\textrm{ss,s}}^{(n_{\textrm{is}})}\right]_{n_{\textrm{ss}},n_{\textrm{hdc}},m_{\textrm{t}}}=\sum_{^{k=1}}^{K}\left[\mathbf{A}_{\textrm{\textrm{df}v}}^{(n_{\textrm{is}})}\right]_{n_{\textrm{ss}},k}\left[\mathbf{A}_{\textrm{\textrm{df}h}}\right]_{n_{\textrm{hdc}},k}\left[\mathbf{D}\right]_{m_{\textrm{t}},k}\nonumber \\
 & =\sum_{^{k=1}}^{K}\left(\left[\mathbf{A}_{\textrm{\textrm{df}v}}^{(1)}\right]_{n_{\textrm{ss}},k}q_{\textrm{ss},k}^{n_{\textrm{is}}-1}\right)\left[\mathbf{A}_{\textrm{\textrm{df}h}}\right]_{n_{\textrm{hdc}},k}\left[\mathbf{D}\right]_{m_{\textrm{t}},k}\nonumber \\
 & =\sum_{^{k=1}}^{K}\left[\mathbf{A}_{\textrm{\textrm{df}v}}^{(1)}\right]_{n_{\textrm{ss}},k}\left[\mathbf{A}_{\textrm{\textrm{df}h}}\right]_{n_{\textrm{hdc}},k}\left(\left[\mathbf{D}\right]_{m_{\textrm{t}},k}q_{\textrm{ss},k}^{n_{\textrm{is}}-1}\right).
\end{align}
Hence,
\begin{align}
\mathcal{Y}_{\textrm{ss,s}}^{(n_{\textrm{is}})}=\left\llbracket \mathcal{Z}_{\textrm{df}};\mathbf{A}_{\textrm{\textrm{df}v}}^{(n_{\textrm{is}})},\mathbf{A}_{\textrm{\textrm{df}h}},\mathbf{D}\right\rrbracket =\left\llbracket \mathcal{Z}_{\textrm{df}};\mathbf{A}_{\textrm{\textrm{df}v}0},\mathbf{A}_{\textrm{\textrm{df}h}},\mathbf{D}^{(n_{\textrm{is}})}\right\rrbracket ,
\end{align}
where $\mathbf{A}_{\textrm{\textrm{df}v}0}=\mathbf{A}_{\textrm{\textrm{df}v}}^{(1)}$
and $\mathbf{D}^{(n_{\textrm{is}})}=\mathbf{D}\mathbf{Q}_{\textrm{ss}}^{n_{\textrm{is}}-1}.$
Therefore, \eqref{eq:smoothingModel} can be rewritten as
\begin{align}
\mathcal{Y}_{\textrm{ss}}=\left[\underset{{\scriptstyle n_{\textrm{is}}=1,\ldots,N_{\textrm{is}}}}{\sqcup_{3}}\mathcal{Y}_{\textrm{ss}}^{(n_{\textrm{is}})}\right]=\left\llbracket \mathcal{Z}_{\textrm{df}};\mathbf{A}_{\textrm{\textrm{df}v}0},\mathbf{A}_{\textrm{\textrm{df}h}},\mathbf{D}_{\textrm{ss}}\right\rrbracket +\mathcal{N}_{\textrm{ss}},\label{eq:modelFinal}
\end{align}
where $\mathbf{D}_{\textrm{ss}}=\left[\left(\mathbf{D}^{(1)}\right)^{T},\left(\mathbf{D}^{(2)}\right)^{T},\ldots,\left(\mathbf{D}^{(N_{\textrm{is}})}\right)^{T}\right]^{T}\in\mathbb{C}^{\left(M_{\textrm{t}}N_{\textrm{is}}\right)\times K},$
and $\mathcal{N}_{\textrm{ss}}=\left[\underset{{\scriptstyle n_{\textrm{is}}=1,\ldots,N_{\textrm{is}}}}{\sqcup_{3}}\mathcal{N}^{(n_{\textrm{is}})}\right]\in\mathbb{C}^{N_{\textrm{ss}}\times N_{\textrm{hdc}}\times\left(M_{\textrm{t}}N_{\textrm{is}}\right)}$
is a tensor of all zeros except $\mathbf{\sigma}_{\textrm{n}}^{2}$
at $(n_{\textrm{ss}},0,\tilde{m}_{\textrm{t}}),$ where $n_{\textrm{ss}}=(N_{\textrm{vd}}N_{\textrm{vs}}-n_{\textrm{is}}+1)$,
$M_{\textrm{t}}(n_{\textrm{is}}-1)\leq\tilde{m}_{\textrm{t}}\leq M_{\textrm{t}}n_{\textrm{is}}$,
and $n_{\textrm{is}}=1,\ldots,N_{\textrm{is}}$.

An illustration of \eqref{eq:modelFinal} is shown at the right of
Fig. \ref{fig:SpatialSmoothingTensorRank}, where the recurrence relations
among the divided subtensors in the mode-1 is equivalent to those
in the mode-3. $\mathbf{A}_{\textrm{\textrm{df}v}}^{(n_{\textrm{is}})}$
and $\mathbf{D}^{(n_{\textrm{is}})}$ are the factor matrices \cite{TensorDecompositions1}
of mode-1 and mode-3, respectively. This property can be used to enhance
the $n$-ranks of the signal tensor model.

We consider the extreme case where the received powers of all the
devices are equal, and these powers are steady temporally across all
time frames, i.e., $\mathbf{\sigma}_{\textrm{s},m_{\textrm{t}},k}^{2}=\mathbf{\sigma}_{\textrm{s}}^{2},$
$m_{\textrm{t}}=1,2,\ldots,M_{\textrm{t}}$, and $k=1,2,\ldots,K$.
Then, $\mathbf{D}=\mathbf{\sigma}_{\textrm{s}}^{2}\mathbf{1}_{M_{\textrm{t}}\times K}.$
As a result, $\mathbf{D}_{\textrm{ss}}$ can be rewritten as $\mathbf{D}_{\textrm{ss}}=\mathbf{\sigma}_{\textrm{s}}^{2}\tilde{\mathbf{Q}}_{\textrm{ss}}\otimes\mathbf{1}_{M_{\textrm{t}}},$
where   $\tilde{\mathbf{Q}}_{\textrm{ss}}=[\mathbf{1}_{K},\mathbf{q}_{\textrm{ss},1},\mathbf{q}_{\textrm{ss},2},\ldots,\mathbf{q}_{\textrm{ss},N_{\textrm{is}}-1}]^{T}\in\mathbb{C}^{N_{\textrm{is}}\times K},$
$\mathbf{q}_{\textrm{ss},n'_{\textrm{is}}}=\left[q_{\textrm{ss},1}^{n'_{\textrm{is}}},q_{\textrm{ss},2}^{n'_{\textrm{is}}},\ldots,q_{\textrm{ss},K}^{n'_{\textrm{is}}}\right]^{T}$,
and $n'_{\textrm{is}}=1,2,\ldots,N_{\textrm{is}}-1.$

Because the paths are from different directions, $\tilde{\mathbf{Q}}_{\textrm{ss}}$
is an $N_{\textrm{is}}\times K$ Vandermonde matrix and $\textrm{Rank}(\tilde{\mathbf{Q}}_{\textrm{ss}})=\min\left(N_{\textrm{is}},K\right).$
$\textrm{Rank}(\tilde{\mathbf{Q}}_{\textrm{ss}})=K$ iff $N_{\textrm{is}}\geq K$.
According to Lemma 1, $\textrm{Rank}_{n}(\tilde{\mathcal{Y}}_{\textrm{ss}})=K$,
when $\textrm{Rank}(\mathbf{D}_{\textrm{ss}})=K$. Thus, the signal
and noise spaces of $\mathcal{Y}_{\textrm{ss}}$ in \eqref{eq:smoothingModel}
can be decomposed in each mode.

We note that the number of $\mathcal{Y}_{\textrm{ss}}^{(n_{\textrm{is}})}$
needs to be larger than the number of devices, i.e., $N_{\textrm{is}}\geq K$,
to guarantee that $\mathcal{Y}_{\textrm{ss}}$ is full rank. Also,
the system DoF available after spatial smoothing is proportional to
the size of $\mathcal{Y}_{\textrm{ss}}^{(n_{\textrm{is}})}$. Since
the total number of elements in $\mathcal{Y}_{\textrm{df}}$ is constant,
increasing the number of $\mathcal{Y}_{\textrm{ss}}^{(n_{\textrm{is}})}$
implies that the size of each $\mathcal{Y}_{\textrm{ss}}^{(n_{\textrm{is}})}$
is smaller, while a larger size of each $\mathcal{Y}_{\textrm{ss}}^{(n_{\textrm{is}})}$
means there is a smaller number of recurrence shifts available. In
this sense, the best strategy is to minimize the difference between
$N_{\textrm{is}}$ and $N_{\textrm{ss}}$. Since in our system, we
have $N_{\textrm{is}}+N_{\textrm{ss}}-1=N_{\textrm{vdc}}=2N_{\textrm{vd}}N_{\textrm{vs}}-1$,
we set $N_{\textrm{ss}}=N_{\textrm{is}}=N_{\textrm{vd}}N_{\textrm{vs}}.$

\noindent \textbf{Remark: }After spatial smoothing, the system DoF
becomes half of that in \eqref{eq:ops}, because we divide $\mathcal{Y}_{\textrm{df}}$
into multiple $\mathcal{Y}_{\textrm{ss}}^{(n_{\textrm{is}})}$. Therefore,
to prevent the system DoF from decreasing and achieve the target set
in Section III-C, we modify (C1) in the optimization problem \eqref{eq:ops}
to $N_{\textrm{vd}}N_{\textrm{vs}}\geq M_{\textrm{vr}}.$ Applying
the analytical strategy developed in Section III-C, we formulate the
modified optimization problem \eqref{eq:ops} as
\begin{align}
\min_{N_{\textrm{vd}},N_{\textrm{hd}}\in\mathbb{Z}^{+}}M_{\textrm{rf}} & =N_{\textrm{vd}}N_{\textrm{hd}}+\left\lceil \frac{M_{\textrm{vr}}}{N_{\textrm{vd}}}\right\rceil \left\lceil \frac{M_{\textrm{hr}}}{N_{\textrm{hd}}}\right\rceil \label{eq:ops-1-1}\\
\textrm{s.t. } & \textrm{(C3) and (C4)}.\nonumber
\end{align}
We see that the minimum $M_{\textrm{rf}}$ can be obtained when $N_{\textrm{vd}}N_{\textrm{hd}}$
and $\left\lceil M_{\textrm{vr}}/N_{\textrm{vd}}\right\rceil \left\lceil M_{\textrm{hr}}/N_{\textrm{hd}}\right\rceil $
are close or equal. Because $N_{\textrm{vd}},N_{\textrm{hd}},N_{\textrm{vs}},N_{\textrm{hs}}\in\mathbb{Z}^{+}$,
the optimal value of $N_{\textrm{vd}},N_{\textrm{hd}},N_{\textrm{vs}},$
and $N_{\textrm{hs}}$ can be obtained.

\section{2-D DoA estimation}

In this section, the 2-D DoAs are estimated by developing a new tensor-based
subspace estimation algorithm. By exploiting the recurrence relations
among the UCAs, the elevation DoAs are estimated first, and then the
corresponding azimuth angles are estimated by using the tensor MUSIC.
The hardware and software complexities are analyzed in the end.

\subsection{Estimation of Elevation Angle }

We first propose a tensor-based total least-squares (TLS)-ESPRIT algorithm
to estimate the elevation angle of each device. The
HOSVD of the measurement tensor $\mathcal{Y}_{\textrm{ss}}$ is given
by
\begin{align}
\mathcal{Y}_{\textrm{ss}} & =\mathcal{L}\times_{1}\mathbf{U}_{\textrm{dfv0}}\times_{2}\mathbf{U}_{\textrm{dfh}}\times_{3}\mathbf{U}_{\textrm{ss}}\nonumber \\
 & =\left\llbracket \mathcal{L};\mathbf{U}_{\textrm{dfv0}},\mathbf{U}_{\textrm{dfh}},\mathbf{U}_{\textrm{ss}}\right\rrbracket \in\mathbb{C}^{N_{\textrm{ss}}\times N_{\textrm{hdc}}\times\left(M_{\textrm{t}}N_{\textrm{is}}\right)},\label{eq:hosvd}
\end{align}
where the unitary matrices, $\mathbf{U}_{\textrm{dfv0}}\in\mathbb{C}^{N_{\textrm{ss}}\times N_{\textrm{ss}}}$,
$\mathbf{U}_{\textrm{dfh}}\in\mathbb{C}^{N_{\textrm{hdc}}\times N_{\textrm{hdc}}}$,
and $\mathbf{U}_{\textrm{ss}}\in\mathbb{C}^{\left(M_{\textrm{t}}N_{\textrm{is}}\right)\times\left(M_{\textrm{t}}N_{\textrm{is}}\right)}$,
are the left singular matrices of the mode-$n$  unfoldings of tensor
$\mathcal{\mathcal{Y}_{\textrm{ss}}},$ and the core tensor $\mathcal{L}\in\mathbb{C}^{N_{\textrm{ss}}\times N_{\textrm{hdc}}\times\left(M_{\textrm{t}}N_{\textrm{is}}\right)}$
is obtained by moving the singular matrices to the left-hand side
of \eqref{eq:hosvd}:
\begin{equation}
\mathcal{L}=\mathcal{Y}_{\textrm{ss}}\times_{1}\mathbf{U}_{\textrm{dfv0}}^{H}\times_{2}\mathbf{U}_{\textrm{dfh}}^{H}\times_{3}\mathbf{U}_{\textrm{ss}}^{H}.
\end{equation}

Define $\tilde{\mathcal{Y}}_{\textrm{ss}}=\left\llbracket \mathcal{Z}_{\textrm{df}};\mathbf{A}_{\textrm{\textrm{df}v0}},\mathbf{A}_{\textrm{\textrm{df}h}},\mathbf{D}_{\textrm{ss}}\right\rrbracket ,$
which contains the noise-free components of $\mathcal{Y}_{\textrm{ss}}$.
By removing the noise subspace component in each mode, we obtain the
HOSVD model of $\tilde{\mathcal{Y}}_{\textrm{ss}}$, as given by
\begin{equation}
\tilde{\mathcal{Y}}_{\textrm{ss}}=\mathcal{L}_{\textrm{ss}}\times_{1}\mathbf{U}_{\textrm{dfv0,s}}\times_{2}\mathbf{U}_{\textrm{dfh,s}}\times_{3}\mathbf{U}_{\textrm{ss,s}}\in\mathbb{C}^{N_{\textrm{ss}}\times N_{\textrm{hdc}}\times\left(M_{\textrm{t}}N_{\textrm{is}}\right)},\label{eq:truncatedHOSVD}
\end{equation}
where $\mathbf{U}_{\textrm{dfv0,s}}\in\mathbb{C}^{N_{\textrm{ss}}\times K}$,
$\mathbf{U}_{\textrm{dfh,s}}\in\mathbb{C}^{N_{\textrm{hdc}}\times K}$,
and $\mathbf{U}_{\textrm{ss,s}}\in\mathbb{C}^{\left(M_{\textrm{t}}N_{\textrm{is}}\right)\times K}$
are the signal subspaces in the first, second, and third modes, respectively;
and $\mathcal{L}_{\textrm{ss}}\in\mathbb{C}^{K\times K\times K}$
is obtained by discarding insignificant singular values of $\mathcal{Y}_{\textrm{ss}}$
in all the modes.

Define the signal subspace as
\begin{equation}
\mathcal{U}_{\textrm{s}}=\mathcal{L}_{\textrm{ss}}\times_{1}\mathbf{U}_{\textrm{dfv0,s}}\times_{2}\mathbf{U}_{\textrm{dfh,s}}\in\mathbb{C}^{N_{\textrm{ss}}\times N_{\textrm{hdc}}\times K}.\label{eq:Us}
\end{equation}
Because $\tilde{\mathcal{Y}}_{\textrm{ss}}$ can be rewritten as
$\tilde{\mathcal{Y}}_{\textrm{ss}}=\mathcal{A}_{\textrm{ss}}\times_{3}\mathbf{D}_{\textrm{ss}}$
with $\mathcal{A}_{\textrm{ss}}=\mathcal{Z}_{\textrm{df}}\times_{1}\mathbf{A}_{\textrm{\textrm{df}v0}}\times_{2}\mathbf{A}_{\textrm{\textrm{df}h}}$,
we obtain
\begin{equation}
\mathcal{A}_{\textrm{ss}}=\mathcal{U}_{\textrm{s}}\times_{3}\mathbf{D}_{\textrm{ss}}.\label{eq:UsA}
\end{equation}
where $\mathbf{D}_{\textrm{ss}}\in\mathbb{C}^{\left(M_{\textrm{t}}N_{\textrm{is}}\right)\times K}$
is a full column rank matrix. According to the shift-invariance relation
among the subtensors in mode-1, we have
\begin{equation}
\mathcal{A}_{\textrm{ss}}\times_{1}\mathbf{J}_{\textrm{v}2}=\mathcal{A}_{\textrm{ss}}\times_{1}\mathbf{J}_{\textrm{v}1}\times_{3}\mathbf{\mathbf{\Theta}_{\textrm{v}}},\label{eq:AJv}
\end{equation}
where $\mathbf{\mathbf{\Theta}_{\textrm{v}}}=\textrm{diag}\left(e^{-j\frac{2\pi}{\lambda}h\cos(\theta_{1})},\ldots,e^{-j\frac{2\pi}{\lambda}h\cos(\theta_{K})}\right)$,
$\mathbf{J}_{\textrm{v}1}=[\mathbf{I}_{M_{\textrm{vr}}-1},\mathbf{0}_{(M_{\textrm{vr}}-1)\times1}]$,
and $\mathbf{J}_{\textrm{v}2}=[\mathbf{0}_{(M_{\textrm{vr}}-1)\times1},\mathbf{I}_{M_{\textrm{vr}}-1}].$
Let
\begin{equation}
\mathcal{U}_{\textrm{sv}1}=\mathcal{U}_{\textrm{s}}\times_{1}\mathbf{J}_{\textrm{v}1}\textrm{ and }\mathcal{U}_{\textrm{sv}2}=\mathcal{U}_{\textrm{s}}\times_{1}\mathbf{J}_{\textrm{v}2}.\label{eq:USV1USV2}
\end{equation}
By substituting \eqref{eq:UsA} into \eqref{eq:AJv}, we have $\mathcal{U}_{\textrm{sv}2}=\mathcal{U}_{\textrm{sv}1}\times_{3}\mathbf{\Psi}_{\textrm{v}}$,
where $\mathbf{\Psi}_{\textrm{v}}\in\mathbb{C}^{K\times K}$ is a
full rank matrix.
To obtain the estimate of $\mathbf{\Psi}_{\textrm{v}}$,
we define $\mathbf{\Upsilon}_{\textrm{v}}=\left[\mathbf{\Upsilon}_{\textrm{v}1}\quad\mathbf{\Upsilon}_{\textrm{v2}}\right]\in\mathbb{C}^{K\times2K}$.
We now generalize the matrix TLS problem formulation \cite{Arraysignalprocessingconceptsandtechniques}
to the tensor setting, as follows.
\begin{align}
\hat{\mathbf{\Upsilon}}_{\textrm{v}} & =\arg\min_{\mathbf{\Upsilon}_{\textrm{v}}}\left\Vert \mathcal{U}_{\textrm{sv}1}\times_{3}\mathbf{\Upsilon}_{\textrm{v}1}+\mathcal{U}_{\textrm{sv}2}\times_{3}\mathbf{\Upsilon}_{\textrm{v2}}\right\Vert ,\label{eq:tls}\\
 & \quad\textrm{s.t.}\quad\mathbf{\Upsilon}_{\textrm{v}}\mathbf{\Upsilon}_{\textrm{v}}^{H}=\mathbf{I}_{K},\nonumber
\end{align}
which finds a unitary matrix $\mathbf{\Upsilon}_{\textrm{v}}$ with
orthogonal submatrices to $\mathcal{U}_{\textrm{sv}1}$ and $\mathcal{U}_{\textrm{sv}2}$
in mode-3.

The mode-3 unfoldings of $\mathcal{U}_{\textrm{sv}1}$
is given by
\begin{equation}
\mathbf{U}_{\textrm{sv}1}{}_{(3)}=\mathbf{U}_{\textrm{s}}{}_{(3)}\left(\mathbf{J}_{\textrm{v}1}\otimes\mathbf{I}_{N_{\textrm{hdc}}}\right)^{T},
\end{equation}
where $\mathbf{U}_{\textrm{s}}{}_{(3)}\in\mathbb{C}^{K\times M_{\textrm{vr}}M_{\textrm{hr}}M{}_{\textrm{f}}}$
is the mode-3 unfolding of $\mathcal{U}_{\textrm{s}}$. The mode-3
unfoldings of $\mathcal{U}_{\textrm{sv}2}$ can be formulated in the
same way. Since $\left\Vert \mathcal{A}\right\Vert =\left\Vert \mathbf{A}_{(n)}\right\Vert _{\textrm{F}}$
$(n=1,2,\ldots,N)$ \cite{TensorDecompositions1}, we rewrite the
tensor TLS problem \eqref{eq:tls} as
\begin{align}
\hat{\mathbf{\Upsilon}}_{\textrm{v}} & =\arg\min_{\mathbf{\Upsilon}_{\textrm{v}}}\left\Vert \mathbf{\Upsilon}_{\textrm{v}1}\mathbf{U}_{\textrm{s}}{}_{(3)}\left(\mathbf{J}_{\textrm{v}1}\otimes\mathbf{I}_{N_{\textrm{hdc}}}\right)^{T}\right.\nonumber \\
 & \qquad\left.+\mathbf{\Upsilon}_{\textrm{v2}}\mathbf{U}_{\textrm{s}}{}_{(3)}\left(\mathbf{J}_{\textrm{v}2}\otimes\mathbf{I}_{N_{\textrm{hdc}}}\right)^{T}\right\Vert _{\textrm{F}}\nonumber \\
 & =\arg\min_{\mathbf{\Upsilon}_{\textrm{v}}}\left\Vert \mathbf{W}_{\textrm{v}}\mathbf{\Upsilon}_{\textrm{v}}^{T}\right\Vert _{\textrm{F}},
\end{align}
where
\begin{align}
\mathbf{W}_{\textrm{v}} & =\left[\left(\mathbf{J}_{\textrm{v}1}\otimes\mathbf{I}_{N_{\textrm{hdc}}}\right)\mathbf{U}_{\textrm{s}}{}_{(3)}^{T}\quad\left(\mathbf{J}_{\textrm{v}2}\otimes\mathbf{I}_{N_{\textrm{hdc}}}\right)\mathbf{U}_{\textrm{s}}{}_{(3)}^{T}\right]\nonumber \\
 & \in\mathbb{C}^{(N_{\textrm{ss}}-1)N_{\textrm{hdc}}\times2K}.
\end{align}
The SVD of $\mathbf{W}_{\textrm{v}}^{H}\mathbf{W}_{\textrm{v}}$ is
written as $\mathbf{W}_{\textrm{v}}^{H}\mathbf{W}_{\textrm{v}}=\mathbf{\dot{U}}_{\textrm{v}}\mathbf{\dot{\Lambda}}_{\textrm{v}}\mathbf{\dot{V}}_{\textrm{v}},$
where $\mathbf{\dot{U}}_{\textrm{v}}\in\mathbb{C}^{2K\times2K}$ and
$\mathbf{\dot{V}}_{\textrm{v}}\in\mathbb{C}^{2K\times2K}$ are the
left and right singular matrices, respectively; and $\mathbf{\dot{\Lambda}}_{\textrm{v}}\in\mathbb{C}^{2K\times2K}$
contains the singular values. We partition $\mathbf{\dot{U}}_{\textrm{v}}$
into four blocks:
\begin{equation}
\mathbf{\dot{U}}_{\textrm{v}}=\left[\begin{array}{cc}
\mathbf{\dot{U}}_{\textrm{v11}} & \mathbf{\dot{U}}_{\textrm{v12}}\\
\mathbf{\dot{U}}_{\textrm{v21}} & \mathbf{\dot{U}}_{\textrm{v22}}
\end{array}\right]\in\mathbb{C}^{2K\times2K}.\label{eq:TLSfinal}
\end{equation}
Let $\hat{\mathbf{\Upsilon}}_{\textrm{v}1}=\mathbf{\dot{U}}_{\textrm{v12}}^{T}\in\mathbb{C}^{K\times K}$
and $\mathbf{\hat{\Upsilon}}_{\textrm{v}2}=\mathbf{\dot{U}}_{\textrm{v22}}^{T}\in\mathbb{C}^{K\times K}$.
According to the standard TLS \cite{Arraysignalprocessingconceptsandtechniques},
the estimate of $\mathbf{\Psi}_{\textrm{v}}$ is given by $\hat{\mathbf{\Psi}}_{\textrm{v}}=-\hat{\mathbf{\Upsilon}}_{\textrm{v}1}\hat{\mathbf{\Upsilon}}_{\textrm{v}2}^{-1},$
where the $K$ eigenvalues of $\hat{\mathbf{\Psi}}_{\textrm{v}}$,
i.e., $\psi_{\textrm{v,}k}$, $k=1,2,\ldots,K$, are sorted in descending
order. According to the array steering expression \eqref{eq:verticalarray},
the elevation angle of the $k$-th device can be estimated as
\begin{equation}
\hat{\theta}_{k}=\arccos\left(j\lambda\ln(\psi_{\textrm{v,}k})/(2\pi h)\right).\label{theta}
\end{equation}

\subsection{Estimation of Azimuth Angle }

We use the tensor-MUSIC algorithm \cite{tensormmWave} to estimate
the azimuth angle of each device. According to \eqref{eq:truncatedHOSVD},
we can discard the largest $K$ singular values of the mode-$n$ unfoldings
of $\mathcal{Y}_{\textrm{ss}}$ and obtain the noise
subspace in mode-2, $\mathbf{U}_{\textrm{dfh,n}}\in\mathbb{C}^{N_{\textrm{hdc}}\times(N_{\textrm{hdc}}-K)}$.
Then, we generalize the matrix-based MUSIC to the tensor, and the
tensor MUSIC spectrum of the azimuth angle can be defined as
\begin{equation}
\textrm{SP}_{\textrm{MUSIC}}(\Phi)=\left\Vert \mathcal{A}_{\textrm{ss}}\times_{2}\mathbf{U}_{\textrm{dfh,n}}^{H}\right\Vert ^{-2},\label{eq:music}
\end{equation}
where $\Phi=\left[\phi_{1},\phi_{2},\ldots,\phi_{K}\right]$. The
mode-2 unfolding of $\mathcal{A}_{\textrm{ss}}$ can be expressed
as
\begin{equation}
\mathbf{A}_{\textrm{ss}}{}_{\textrm{(2)}}=\mathbf{A}_{\textrm{\textrm{df}h}}\mathbf{Z}_{\textrm{\textrm{df}}(2)}\left(\mathbf{I}_{M_{\textrm{t}}N_{\textrm{is}}}\otimes\mathbf{A}_{\textrm{\textrm{df}v0}}\right)^{T},\label{eq:A2}
\end{equation}
where $\mathbf{Z}_{\textrm{\textrm{df}}(2)}$ is the mode-2 unfolding of $\mathcal{Z}_{\textrm{df}}$.
According to a property of tensor multiplication and unfolding:
$\left\Vert \mathcal{A}\right\Vert =\left\Vert \mathbf{A}_{(n)}\right\Vert _{\textrm{F}},\:n=1,2,\ldots,N$,
we can rewrite the tensor MUSIC spectrum \eqref{eq:music} as
\begin{equation}
\textrm{SP}_{\textrm{MUSIC}}(\Phi)=\left\Vert \mathbf{U}_{\textrm{dfh,n}}^{H}\mathbf{A}_{\textrm{\textrm{df}h}}\mathbf{Z}_{\textrm{\textrm{df}}(2)}\left(\mathbf{I}_{M_{\textrm{t}}N_{\textrm{is}}}\otimes\mathbf{A}_{\textrm{\textrm{df}v0}}\right)^{T}\right\Vert _{\textrm{F}}^{-2}.\label{eq:MUSIC2}
\end{equation}
By substituting the estimated elevation angle of each device \eqref{theta}
into \eqref{eq:MUSIC2}, the corresponding azimuth angle $\phi_{k}$
can be estimated by searching the prominent peaks of the tensor MUSIC
spectrum \eqref{eq:MUSIC2}. Algorithm \ref{alg:algorithmESPRIT}
summarizes the procedure of the proposed tensor-based subspace estimation
algorithm.
\begin{algorithm}
\protect\caption{Tensor-based subspace estimation algorithm\label{alg:algorithmESPRIT}}

\begin{itemize}
\item \textbf{Input}: The processed signal, $\mathcal{Y}_{\textrm{ss}}$,
and the number of devices, $K$.
\item \textbf{Output}: The estimated elevation and azimuth angles, $\hat{\theta}_{k}$
and $\hat{\phi}_{k}$, $k=1,2,\ldots,K$.
\item Take the HOSVD of $\mathcal{Y}_{\textrm{ss}}$ and
obtain $\mathcal{U}_{\textrm{s}}$, $\mathcal{U}_{\textrm{sv}1}$,
and $\mathcal{U}_{\textrm{sv}2}$ according to \eqref{eq:Us} and
\eqref{eq:USV1USV2}.
\item Estimate $\hat{\mathbf{\Psi}}_{\textrm{v}}$ by solving the tenosr
TLS problem \eqref{eq:tls}.
\item Calculate the eigenvalues of $\hat{\mathbf{\Psi}}_{\textrm{v}}$,
i.e., $\psi_{\textrm{v,}k}$, $k=1,2,\ldots,K$, and estimate $\hat{\theta}_{k}$
by using \eqref{theta}.
\item Calculate $\mathbf{U}_{\textrm{dfh,n}}$ and estimate $\hat{\phi}_{k}$
by searching the prominent peaks of \eqref{eq:MUSIC2}.
\end{itemize}
\end{algorithm}

\subsection{Complexity Analysis}

We analyze the hardware and software complexity of the proposed tensor-based
parameter estimation algorithm.
\begin{figure}
\begin{centering}
\includegraphics[width=6cm]{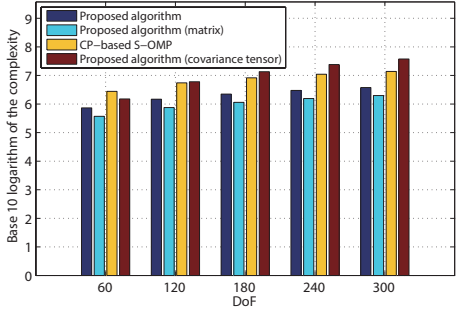}
\par\end{centering}
\centering{}\caption{Variation of the software complexity vs. DoF.\label{fig:Complex}}
\end{figure}

For the hardware complexity, the use of the proposed hybrid array
reduces the hardware complexity to $O(M_{\textrm{rf}})=O(N_{\textrm{vd}}N_{\textrm{hd}}+N_{\textrm{vs}}N_{\textrm{hs}})$,
while a fully digital array using the same number of antennas would
have a hardware complexity of $O(M_{\textrm{bs}})$. We
compare the system power consumption between our system, and the systems
using hourglass arrays \cite{2DSPARSE} and OBAs \cite{OBA}\footnote{For a fair comparison, all these systems do not consider
using spatial smoothing, and the periodicity of UCAs is considered
here.}. According to \cite{PowerConsumption}, the power
of a hybrid array is consumed by its RF chains, analog-to-digital converters (ADCs), local oscillators,
power amplifiers, and phase shifters. Since the proposed method, and
the methods using hourglass arrays and OBAs, are only different in
terms of the design of RF connection matrices, the numbers of required
phase shifters, local oscillators, and power amplifiers are the same
across these three methods. As a result, the difference of system power consumption between  the methods depends on  the numbers of RF chains and ADSs. Also note that the number of RF chains is equal to the number of
ADCs.
Assume that the dimension of phase-shifter output ports  is $M_{\textrm{hr}}\times M_{\textrm{vr}}=29\times17$.
In our system, only $M_{\textrm{rf}}=N_{\textrm{vd}}N_{\textrm{hd}}+N_{\textrm{vs}}N_{\textrm{hs}}-1=32$
RF chains (and ADSs) are required by solving \eqref{eq:ops}. However, in the
systems using hourglass arrays and OBAs, the numbers of required RF
chains are 37 and 35, respectively.

As for the signal processing complexity, we compare the computational
complexity of the proposed tensor-based algorithm with its matrix-based
counterpart, which formulates the signal model in the matrix form
and uses matrix-based ESPRIT-MUSIC algorithm for DoA estimation. For
matrix-based algorithms, the computational complexity of performing
SVD to the measurement sample matrix and truncating its rank to $K$
is $O(N_{\textrm{ss}}N_{\textrm{hdc}}M_{\textrm{t}}N_{\textrm{is}}K)$.
The complexities of estimating the elevation and azimuth angles are
$O(K^{3}+N_{\textrm{ss}}N_{\textrm{hdc}})$ and $O(N_{\textrm{hdc}}K^{2}+N_{\textrm{hdc}}^{2}KD),$
respectively. $D$ is the size of search dimension. For the proposed
tensor-based algorithm, the computational complexity of taking the
HOSVD of the tensor model is $O(N_{\textrm{ss}}N_{\textrm{hdc}}M_{\textrm{t}}N_{\textrm{is}}K)$.
The computational complexities of estimating elevation and azimuth
angles are $O(N_{\textrm{ss}}N_{\textrm{hdc}}+K^{3})$ and $O(N_{\textrm{ss}}N_{\textrm{hdc}}M_{\textrm{t}}N_{\textrm{is}}K+N_{\textrm{hdc}}^{2}KD)$,
respectively. The new tensor-based algorithm needs slightly more computations,
but is in the same order as its matrix-based counterpart.
\begin{figure*}
\begin{centering}
\includegraphics[width=14cm]{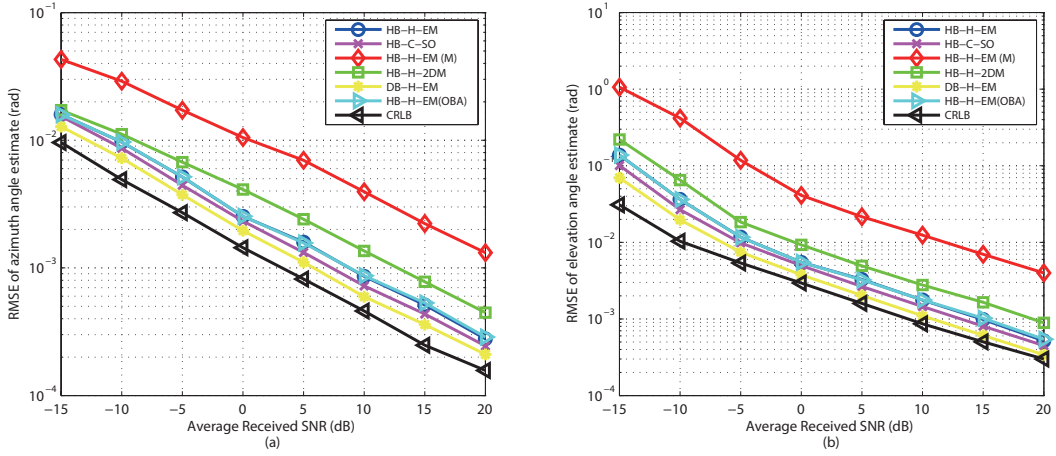}
\par\end{centering}
\centering{}\caption{RMSE vs. the average received SNR for the estimation of DoAs for identifying
$K=50$ devices. (a) Azimuth angle; (b) Elevation angle.\label{fig:snr1}}
\end{figure*}
\begin{figure*}
\begin{centering}
\includegraphics[width=14cm]{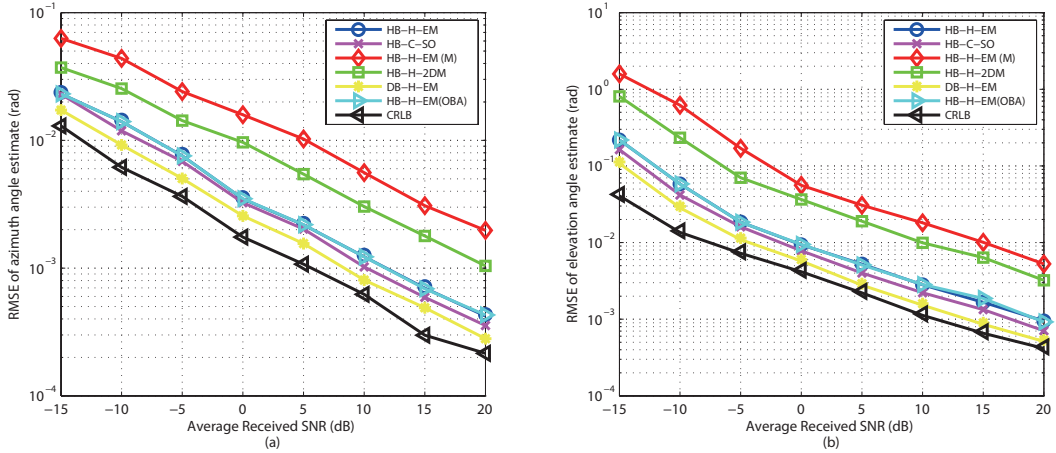}
\par\end{centering}
\centering{}\caption{RMSE vs. the average received SNR for the estimation of DoAs for identifying
$K=200$ devices. (a) Azimuth angle; (b) Elevation angle.\label{fig:snr1-1}}
\end{figure*}

We also compare our algorithm with the CP-based simultaneous-orthogonal
matching pursuit (S-OMP) algorithm \cite{tensormmWave}. The algorithm
first applies CP decomposition to decompose the received signal tensor
model, and then applies S-OMP to estimate the parameters. The complexities
of the CP decomposition and S-OMP are $O(N_{\textrm{ss}}N_{\textrm{hdc}}M_{\textrm{t}}N_{\textrm{is}}K+N_{\textrm{ss}}N_{\textrm{hdc}}K^{2}+K^{3})$
and $O(N_{\textrm{ss}}N_{\textrm{hdc}}M_{\textrm{t}}N_{\textrm{is}}(N_{1}+N_{2}))$,
respectively, where $N_{1}\gg K$ and $N_{2}\gg K$ are the dimensions
of the OMP grid. The complexity of CP-based subspace algorithm is
much higher than that of our HOSVD-based algorithm.

Note that all the operations in our DoA estimation
algorithm are on the signal data tensor model directly. If our algorithms
operate on the signal covariance tensor, we need to calculate the signal covariance tensor model \cite{Tensor2DMUSIC,tensorMUSIC,tensorBFmmWave}
\begin{equation}
\mathcal{R}_{\textrm{ss}}=\frac{1}{M_{\textrm{t}}N_{\textrm{is}}}\sum_{m=1}^{M_{\textrm{t}}N_{\textrm{is}}}\mathcal{Y}_{\textrm{ss},m}\circ\mathcal{Y}_{\textrm{ss},m}^{*}\in\mathbb{C}^{N_{\textrm{ss}}\times N_{\textrm{hdc}}\times N_{\textrm{ss}}\times N_{\textrm{hdc}}},\label{eq:R}
\end{equation}
where $\mathbf{Y}_{\textrm{ss},m}\in\mathbb{C}^{N_{\textrm{ss}}\times N_{\textrm{hdc}}}$
is the $m$-th subtensor of $\mathcal{Y}_{\textrm{ss}}$, $m=1,2,\ldots,M_{\textrm{t}}N_{\textrm{is}}$,
and then, take the HOSVD of \eqref{eq:R}. The computational complexity
of this process is $O(M_{\textrm{t}}N_{\textrm{is}}N_{\textrm{ss}}^{2}N_{\textrm{hdc}}^{2}+N_{\textrm{ss}}^{2}N_{\textrm{hdc}}^{2}K)$,
which needs much more computations than our algorithms.

The results of the computational complexities of our algorithm, its
matrix-based counterpart, CP-based S-OMP algorithm, and the proposed
algorithm operating on the covariance tensor, as a function of the
system DoF, $O(N_{\textrm{ss}}N_{\textrm{hdc}})$, are presented in
Fig. \ref{fig:Complex}, where $M_{\textrm{t}}=20$, $N_{\textrm{is}}=20$,
$K=15$, and $N_{1}=N_{2}=50$. The figure shows that our proposed
algorithm requires more computations than its matrix-based counterpart
at the gain of significantly improved DoA estimation performance,
as will be seen in Section VI. However, compared to the other two
algorithms, the computational complexity of our algorithm is much
lower.

\section{Simulation Results}

In this section, simulation results are provided to demonstrate the
performance of the proposed algorithm. The system bandwidth is $B=1$
GHz. The number of time frames is set to $M_{\textrm{t}}$= 20. The
reference radial frequency $f=28$ GHz. The vertical spacing between
adjacent receiving UCAs is $h=0.5\lambda$ and the radius of the UCyA
is $r=2\lambda$, where $\lambda=c/f$ and $c$ is the speed of light. The geometry parameters of
the UCyA are $M_{\textrm{v}}=25$ and $M_{\textrm{h}}=30$. For the
hybrid beamforming, we set $N_{\textrm{vd}}=5$, $N_{\textrm{hd}}=5$,
$N_{\textrm{vs}}=5$, and $N_{\textrm{hs}}=6$, so there are $M_{\textrm{rf}}=N_{\textrm{vd}}N_{\textrm{hd}}+N_{\textrm{vs}}N_{\textrm{hs}}-1=54$
RF chains in our system.

Fig. \ref{fig:snr1} plots the root mean square errors (RMSEs) for
the estimates of the azimuth and elevation angles versus the average
received SNR, where the DoAs of $K=50$ devices are estimated. By
using the proposed nested sparse hybrid beamforming, we compare the
proposed HOSVD-based ESPRIT-MUSIC (HB-H-EM) algorithm with its reduced
version in the matrix form (HB-H-EM (M)), the CP-based S-OMP (HB-C-SO)
algorithm \cite{tensormmWave}, the HOSVD-based 2D MUSIC (HB-H-2DM)
algorithm \cite{Tensor2DMUSIC}, and the proposed
algorithm but using OBA to design the RF connection matrix (HB-H-EM
(OBA)). We also apply the proposed algorithm for fully digital beamforming
(DB-H-EM), and provide the  Cram\'{e}r-Rao lower bound (CRLB) \cite{bounds}
as a reference. We can see that all the estimated algorithms approach
the CRLB, as the average received SNR increases. Fig. \ref{fig:snr1}
also shows that our proposed tensor-based algorithm provides a better
accuracy than its matrix-based counterparts. This is because the tensor-based
algorithm can suppress the noise components in each mode of the signal
tensor model, while the matrix-based algorithm can only suppress the
noise in the time domain corresponding to the third mode in this paper.
By applying CP to decompose the signal tensor model, HB-C-SO achieves
better estimation performance than other HOSVD-like algorithms. However,
the performance improvement is limited because HB-C-SO uses S-OMP
to estimate the parameters, generating quantized estimates only. We
also observe that the precision of the angle estimation of our proposed
algorithm is a bit lower than that of DB-H-EM. However, both DB-H-EM
and HB-C-SO have a much higher complexity than our algorithm, as analyzed
in  Section V-C. In addition, the estimation accuracy
is nearly the same between the proposed HB-H-EM and HB-H-EM (OBA).
This is because the DoA estimation accuracy depends on the dimension
of the difference coarray, not the dimension of the RF chain network,
while the constructed difference coarrays of the two methods are identical.

\begin{figure}
\begin{centering}
\includegraphics[width=7cm]{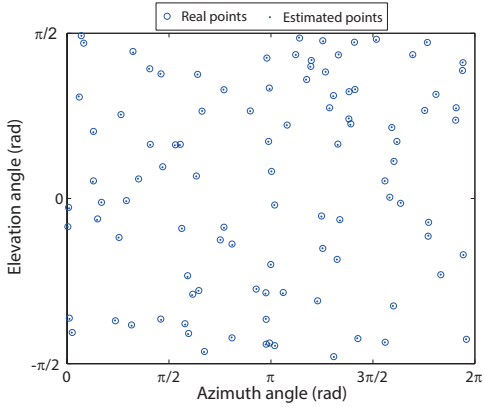}
\par\end{centering}
\centering{}\caption{2-D DoA estimation by using the proposed algorithm for $100$ devices.\label{fig:2-D-DOA-estimation}}
\end{figure}

Fig. \ref{fig:snr1-1} shows the RMSEs for the estimates of DoAs versus
the average received SNR. The number of devices is $K=200$.
By comparing Figs. \ref{fig:snr1} and \ref{fig:snr1-1}, we see that
for the fixed SNR and a fixed number of RF chains, the estimation
accuracy of all the tested algorithms decreases as $K$ increases.
This is because as $K$ grows, more signal components need to be estimated
and distinguished. Compared with Fig. \ref{fig:snr1}, Fig. \ref{fig:snr1-1}
also shows that the performance degeneration of HB-H-2DM is larger
than other algorithms. This is because HB-H-2DM uses signal covariance
tensors for the 2-D DoA estimation, and its MUSIC spectrum is a product
of multiple separable second-order mode-$n$ spectra, which results
in undesirable cross-terms \cite{Tensor2DMUSIC} and compromises the
estimation accuracy.

Fig. \ref{fig:2-D-DOA-estimation} evaluates the performance of our
proposed algorithm. Without loss of generality, here we estimate the
2-D DoAs of $K=100$ devices, where $\textrm{SNR}=5$ dB. As seen
from the results, the proposed algorithm can accurately estimate the
azimuth and elevation angles of $100$ devices. All the estimates
are well matched with the actual values, while we only use 54 RF chains
in our system.

\section{Conclusion}

We presented a novel sparse nested hybrid UCyA for mIoT networks.
By exploiting the difference coarray technique and tailoring for the
UCyA, we proposed a channel estimation scheme based on the second-order
statistics of the received signals. As a result, the designed hybrid
array only requires a small number of RF chains to achieve DoA estimation
for a massive number of IoT devices. We proposed a spatial smoothing-based
method to enhance the $n$-ranks of the signal tensor model. By using
the method, a large enough rank in each mode of the signal tensor
model was provided for the DoA estimation of $K$ devices. Given the
designed hybrid array, a new tensor-based 2-D DoA estimation algorithm
was proposed, and it can significantly improve the estimation accuracy
while reducing the computational complexity. Simulation results demonstrated
that our proposed hybrid array system can accurately estimate the
2-D DoAs of a large number of IoT devices.

\section*{Appendix I\protect \protect \protect \protect \protect \protect
\protect \protect }

\section*{Proof of Theorem 1}

Let $\gamma(\theta_{k})=2\pi r\sin(\theta_{k})/\lambda.$ The
phase-space transformation of   $\mathbf{a}_{\textrm{h}}(\theta_{k},\phi_{k})$
can be expressed as
\begin{align}
 & a_{\textrm{hps},p}(\theta_{k},\phi_{k})=\sum_{m_{\textrm{h}}=1}^{M_{\textrm{h}}}\left(a_{\textrm{h},m_{\textrm{h}}}(\theta_{k},\phi_{k})\right)e^{-j\frac{2\pi(m_{\textrm{h}}-1)}{M_{\textrm{h}}}p}\nonumber \\
 & =\sum_{m_{\textrm{h}}=1}^{M_{\textrm{h}}}\left(\frac{1}{\sqrt{M_{\textrm{h}}}}e^{j\gamma(\theta_{k})\cos(\phi_{k}-\varphi_{m_{\textrm{h}}})}\right)e^{-j\frac{2\pi(m_{\textrm{h}}-1)}{M_{\textrm{h}}}p}\nonumber \\
 & \stackrel{(\textrm{a})}{=}\sum_{m_{\textrm{h}}=1}^{M_{\textrm{h}}}\left(\frac{1}{\sqrt{M_{\textrm{h}}}}\sum_{q=-\infty}^{\infty}j^{q}J_{q}\left(\gamma(\theta_{k})\right)e^{jq(\phi_{k}-\varphi_{m_{\textrm{h}}})}\right)\nonumber \\
 & \qquad\times e^{-j\frac{2\pi(m_{\textrm{h}}-1)}{M_{\textrm{h}}}p}\nonumber \\
 & \stackrel{(\textrm{b})}{=}\frac{1}{\sqrt{M_{\textrm{h}}}}\sum_{Q=-\infty}^{\infty}M_{\textrm{h}}j^{(QM_{\textrm{h}}-p)}J_{(QM_{\textrm{h}}-p)}\left(\gamma(\theta_{k})\right)\nonumber \\
 & \qquad\times e^{j(QM_{\textrm{h}}-p)\phi_{k}}\nonumber \\
 & \stackrel{(\textrm{c})}{=}\sqrt{M_{\textrm{h}}}\left[\vphantom{\sum_{Q=-\infty,Q\neq0}^{\infty}}j^{p}J_{p}\left(\gamma(\theta_{k})\right)e^{-jp\phi_{k}}\right.\nonumber \\
 & \qquad\left.+\sum_{Q=-\infty,Q\neq0}^{\infty}\varepsilon_{p,Q}\left(\gamma(\theta_{k}),\phi_{k}\right)\right]\label{eq:T1}
\end{align}
where
\begin{align*}
 & \varepsilon_{p,Q}\left(\gamma(\theta_{k}),\phi_{k}\right)=j^{(QM_{\textrm{h}}-p)}J_{(QM_{\textrm{h}}-p)}\left(\gamma(\theta_{k})\right)e^{j(QM_{\textrm{h}}-p)\phi_{k}}.
\end{align*}
In \eqref{eq:T1}, $(\textrm{a})$ and $(\textrm{c})$ follow the
important properties of the Bessel function, i.e., $e^{jx\cos y}=\sum_{v=-\infty}^{\infty}j^{v}J_{v}(x)e^{jvy}$
and $J_{-v}(x)=(-1)^{v}J_{v}(x)$, respectively. $(\textrm{b})$ is
obtained by letting $p+q=QM_{\textrm{h}}$ \cite{bessel}.

Let $x=v\rho,\rho\in(0,1)$ and $v\in\mathbb{Z}^{+}$. The Bessel function,
$J_{v}(x)$, whose order $v$ exceeds its argument, $x$, can be written
in the following form \cite{bessel}

\begin{equation}
J_{v}(v\rho)=\frac{1}{\pi}\int_{0}^{\pi}\exp\left(-vF(\vartheta,\rho)\right)d\vartheta,\label{eq:bessel}
\end{equation}
where
\[
F(\vartheta,\rho)=\log\left(\frac{\vartheta+\sqrt{\vartheta^{2}-\rho^{2}\sin^{2}\vartheta}}{\rho\sin\vartheta}\right)-\cot\vartheta\sqrt{\vartheta^{2}-\rho^{2}\sin^{2}\vartheta}.
\]

The partial derivative of \eqref{eq:bessel} with
respect to $\rho$ is given by

\begin{align}
\frac{\partial}{\partial\rho}J_{v}(v\rho)=-\frac{v}{\pi}\int_{0}^{\pi}\frac{\partial F(\vartheta,\rho)}{\partial\rho}\exp\left(-vF(\vartheta,\rho)\right)d\vartheta\nonumber \\
=\frac{v}{\pi\rho}\int_{0}^{\pi}g(\vartheta,\rho)\exp\left(-vF(\vartheta,\rho)\right)d\vartheta,\label{eq:besselpartial}
\end{align}
where $g(\vartheta,\rho)=\left(\vartheta-\rho^{2}\sin\vartheta\cos\vartheta\right)/\sqrt{\vartheta^{2}-\rho^{2}\sin^{2}\vartheta}.$
Given that
\begin{align}
g(\vartheta,\rho) & =\frac{\vartheta-\rho^{2}\sin\vartheta\cos\vartheta}{\sqrt{\vartheta^{2}-\rho^{2}\sin^{2}\vartheta}}\geq\frac{\vartheta-\sin\vartheta\cos\vartheta}{\sqrt{\vartheta^{2}-\rho^{2}\sin^{2}\vartheta}}\nonumber \\
 & \geq\frac{\vartheta-\sin\vartheta}{\sqrt{\vartheta^{2}-\rho^{2}\sin^{2}\vartheta}}\geq0,
\end{align}
we have $\partial J_{v}(v\rho)/\partial\rho>0,$ and conclude that
$J_{v}(v\rho)$ is an increasing function of $\rho$. Thus, $J_{v}(v\rho)<J_{v}(v).$

On the other hand, the partial derivative of \eqref{eq:bessel}
with respect to $v$ is given by
\begin{equation}
\frac{\partial}{\partial v}J_{v}(v\rho)=-\frac{1}{\pi}\int_{0}^{\pi}F(\vartheta,\rho)\exp\left(-vF(\vartheta,\rho)\right)d\vartheta.
\end{equation}
Because
\begin{equation}
\frac{\partial}{\partial\vartheta}F(\vartheta,\rho)=\frac{(1-\rho\cot\vartheta)^{2}}{\sqrt{\vartheta^{2}-\rho^{2}\sin^{2}\vartheta}}+\sqrt{\vartheta^{2}-\rho^{2}\sin^{2}\vartheta}\geq0
\end{equation}
and $\partial F(0,\rho)/\partial\rho=-\sqrt{1-\rho^{2}}/\rho\leq0$,
we have $F(\vartheta,\rho)\geq F(0,\rho)\geq F(0,1)=0$ and hence,
$\partial J_{v}(v\rho)/\partial v<0$. This means that $J_{v}(v\rho)$
is a decreasing function of $v$, i.e., $J_{v}(v\rho)<J_{1}(\rho).$
Therefore, we have $J_{v}(v\rho)<J_{v}(v)<J_{1}(1)\approx0.4$ with
$\rho\in(0,1)$ and $v\in\mathbb{Z}^{+}$. For $\left|v\right|>\left|x\right|$,
$\left|J_{v}(x)\right|\approx0$ with $v\in\mathbb{Z}^{+}$. Based
on this property, , both $\varepsilon_{p,Q}\left(\gamma(\theta_{k}),\phi_{k}\right)$
and $J_{p}\left(\gamma(\theta_{k})\right)$ in \eqref{eq:T1} can
be suppressed in the case of $\left|p\right|>P>\gamma(\theta_{k})$,
since $P\geq\left\lfloor 2\pi r/\lambda\right\rfloor >2$ and $M_{\textrm{h}}\geq\left\lfloor 4\pi r/\lambda\right\rfloor >2$.
When $\left|p\right|\leq P$, we can only ignore $\varepsilon_{p,Q}\left(\gamma(\theta_{k}),\phi_{k}\right)$.
Thus, \eqref{eq:T1} can be approximated by \eqref{eq:Theorem1}.
This concludes the proof.

\section*{Appendix II\protect \protect \protect \protect \protect \protect
\protect \protect }

\section*{Proof of Theorem 2}

Define $\mathcal{Y}_{\textrm{dfs}}=\left\llbracket \mathcal{Z}_{\textrm{df}};\mathbf{A}_{\textrm{\textrm{df}v}},\mathbf{A}_{\textrm{\textrm{df}h}},\mathbf{D}\right\rrbracket $,
which is the noise-free model of $\mathcal{Y}_{\textrm{df}}$. Thus,
$\mathcal{Y}_{\textrm{dfs}}$ consists of all the signal space components.

Because $\mathbf{A}_{\textrm{\textrm{df}v}}\in\mathbb{C}^{N_{\textrm{vdc}}\times K}$
and $\mathbf{A}_{\textrm{\textrm{df}h}}\in\mathbb{C}^{N_{\textrm{hdc}}\times K}$
are Vandermonde matrices, and in our system, we have $N_{\textrm{vdc}}\geq K$
and $N_{\textrm{hdc}}\geq K$, according to uniqueness condition of
the CP decomposition, the $n$-ranks of $\mathcal{Y}_{\textrm{dfs}}$
depends on the rank of $\mathbf{D}$.

On the other hand, the SVD of the mode-$n$ unfolding of $\mathcal{Y}_{\textrm{dfs}}$,
$\mathbf{Y}_{\textrm{dfs}}{}_{(n)}$, can be written as $\mathbf{Y}_{\textrm{dfs}}{}_{(n)}=\mathbf{U}_{\textrm{vs},n}\mathbf{\mathbf{\Sigma}}_{\textrm{vs},n}\mathbf{\mathbf{V}}_{\textrm{vs},n}^{H},$
where $n=1,2,3$, and we have $\textrm{Rank}(\mathbf{Y}_{\textrm{dfs}}{}_{(n)})=\textrm{Rank}(\mathbf{U}_{\textrm{vs},n})=\textrm{Rank}(\mathbf{\mathbf{\Sigma}}_{\textrm{vs},n})=\textrm{Rank}(\mathbf{\mathbf{V}}_{\textrm{vs},n}).$
When $\textrm{Rank}(\mathbf{D})<K$, we have $\textrm{Rank}(\mathbf{Y}_{\textrm{dfs}}{}_{(n)})<K,$
and thus $\textrm{Rank}(\mathbf{U}_{\textrm{v,}n})<K.$

This concludes the proof of Theorem 2.

\bibliographystyle{IEEEbib}

\begin{IEEEbiography}[{\includegraphics[width=1in,height=1.25in,clip,keepaspectratio]{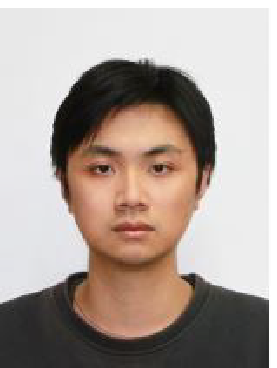}}]{Zhipeng Lin}
(S'17) is currently working toward the dual Ph.D. degrees in communication and information engineering with the School of Information and Communication Engineering, Beijing University of Posts and Telecommunications, Beijing, China, and the School of Electrical and Data Engineering, University of Technology of Sydney, Sydney, NSW, Australia. His current research interests include millimeter-wave communication, massive MIMO, hybrid beamforming, wireless localization, and tensor processing.
\end{IEEEbiography}

\begin{IEEEbiography}[{\includegraphics[width=1in,height=1.25in,clip,keepaspectratio]{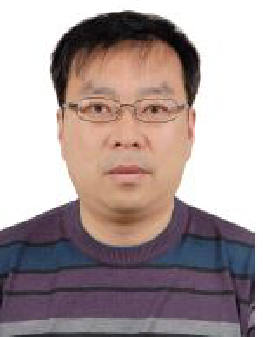}}]{Tiejun Lv}
(M'08-SM'12) received the M.S. and Ph.D. degrees in electronic engineering from the University of Electronic Science and Technology of China (UESTC), Chengdu, China, in 1997 and 2000, respectively. From January 2001 to January 2003, he was a Postdoctoral Fellow with Tsinghua University, Beijing, China. In 2005, he was promoted to a Full Professor with the School of Information and Communication Engineering, Beijing University of Posts and Telecommunications (BUPT). From September 2008 to March 2009, he was a Visiting Professor with the Department of Electrical Engineering, Stanford University, Stanford, CA, USA. He is the author of 3 books, more than 80 published IEEE journal papers and 180 conference papers on the physical layer of wireless mobile communications. His current research interests include signal processing, communications theory and networking. He was the recipient of the Program for New Century Excellent Talents in University Award from the Ministry of Education, China, in 2006. He received the Nature Science Award in the Ministry of Education of China for the hierarchical cooperative communication theory and technologies in 2015.
\end{IEEEbiography}

\begin{IEEEbiography}[{\includegraphics[width=1in,height =1.25in,clip,keepaspectratio]{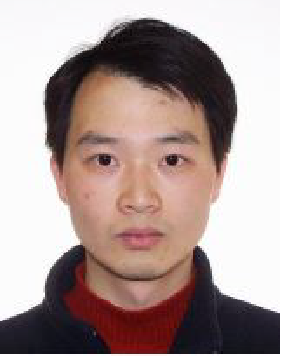}}]{Wei Ni}
(M'09-SM'15) received the B.E. and Ph.D. degrees in Electronic Engineering from Fudan University, Shanghai, China, in 2000 and 2005, respectively. Currently, he is a Group Leader and Principal Research Scientist at CSIRO, Sydney, Australia, and an Adjunct Professor at the University of Technology Sydney and Honorary Professor at Macquarie University, Sydney. He was a Postdoctoral Research Fellow at Shanghai Jiaotong University from 2005 -- 2008; Deputy Project Manager at the Bell Labs, Alcatel/Alcatel-Lucent from 2005 to 2008; and Senior Researcher at Devices R\&D, Nokia from 2008 to 2009. His research interests include signal processing, stochastic optimization, learning, as well as their applications to network efficiency and integrity.

Dr Ni is the Chair of IEEE Vehicular Technology Society (VTS) New South Wales (NSW) Chapter since 2020 and an Editor of IEEE Transactions on Wireless Communications since 2018. He served first the Secretary and then Vice-Chair of IEEE NSW VTS Chapter from  2015 to 2019, Track Chair for VTC-Spring 2017, Track Co-chair for IEEE VTC-Spring 2016, Publication Chair for BodyNet 2015, and  Student Travel Grant Chair for WPMC 2014.
\end{IEEEbiography}

\begin{IEEEbiography}[{\includegraphics[width=1in,height=1.25in,clip,keepaspectratio]{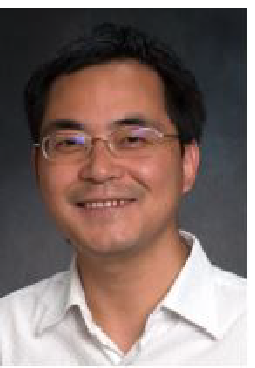}}]{J. Andrew Zhang}
(M'04-SM'11) received the B.Sc. degree from Xi'an JiaoTong University, China, in 1996, the M.Sc. degree from Nanjing University of Posts and Telecommunications, China, in 1999, and the Ph.D. degree from the Australian National University, in 2004.

Currently, Dr. Zhang is an Associate Professor in the School of Electrical and Data Engineering, University of Technology Sydney, Australia. He was a researcher with Data61, CSIRO, Australia from 2010 to 2016, the Networked Systems, NICTA, Australia from 2004 to 2010, and ZTE Corp., Nanjing, China from 1999 to 2001.  Dr. Zhang's research interests are in the area of signal processing for wireless communications and sensing. He has published more than 180 papers in leading international Journals and conference proceedings, and has won 5 best paper awards. He is a recipient of CSIRO Chairman's Medal and the Australian Engineering Innovation Award in 2012 for exceptional research achievements in multi-gigabit wireless communications.
\end{IEEEbiography}

\begin{IEEEbiography}[{\includegraphics[width=1in,height=1.25in,clip,keepaspectratio]{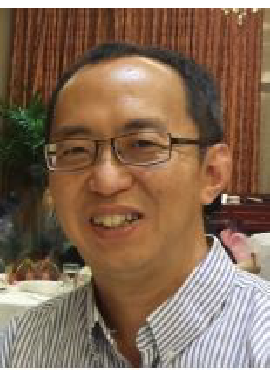}}]{Ren Ping Liu}
(M'09-SM'14) received his B.E. and M.E. degrees from Beijing University of Posts and Telecommunications, China, and the Ph.D. degree from the University of Newcastle, Australia.

He is currently a Professor and Head of Discipline of Network \& Cybersecurity at University of Technology Sydney. Professor Liu was the co-founder and CTO of Ultimo Digital Technologies Pty Ltd, developing IoT and Blockchain. Prior to that he was a Principal Scientist and Research Leader at CSIRO, where he led wireless networking research activities. He specialises in system design and modelling and has delivered networking solutions to a number of government agencies and industry customers. His research interests include wireless networking, Cybersecurity, and Blockchain.

Professor Liu was the founding chair of IEEE NSW VTS Chapter and a Senior Member of IEEE. He served as Technical Program Committee chairs and Organising Committee chairs in a number of IEEE Conferences. Prof Liu was the winner of Australian Engineering Innovation Award and CSIRO Chairman medal. He has over 200 research publications.
\end{IEEEbiography}

\end{document}